\documentclass[11pt,english]{article} 
\pdfoutput=1 
\usepackage{endnotes}
\let\footnote=\endnote
\let\enotesize=\normalsize
\def\notesname{Endnotes}%
\def\makeenmark{$^{\theenmark}$}
\def\enoteformat{\rightskip0pt\leftskip0pt\parindent=1.75em
  \leavevmode\llap{\theenmark.\enskip}}

\usepackage{natbib}
 \bibpunct[, ]{(}{)}{,}{a}{}{,}%
 \def\bibfont{\small}%
 \def\bibsep{\smallskipamount}%
 \def\bibhang{24pt}%

\usepackage{graphicx}
\usepackage{wrapfig}
\usepackage{amsmath,amssymb}
\usepackage{authblk}
\usepackage{subfigure}
\usepackage{algorithm}
\usepackage{setspace}
\usepackage{bbm}
\usepackage[T1]{fontenc}
\usepackage[utf8]{inputenc}
\usepackage{babel}

\usepackage{amsthm}
\usepackage{color}
\usepackage{setspace}
\usepackage{url}
\usepackage{subfigure}
\usepackage[margin=1in,letterpaper]{geometry}
\usepackage{paralist}
\usepackage{titlesec}
\usepackage{algpseudocode}
\usepackage{booktabs}
\usepackage{authblk}

\usepackage{mathtools}

\newcommand{\defn}{\mathrel{\mathop:}=}
\newcommand{\real}{{\rm{I\hspace{-.75mm}R}}}

\newcommand{\bX}{{\bf X}}
\newcommand{\bZ}{{\bf Z}}
\newcommand{\bW}{{\bf W}}
\newcommand{\tbZ}{\tilde{{\bf Z}}}
\newcommand{\bz}{{\bf z}}
\newcommand{\tbz}{\tilde{{\bf z}}}
\newcommand{\bR}{{\bf R}}

\newcommand{\bx}{{\bf x}}

\newcommand{\bPhi}{{\bf \Phi}}
\newcommand{\bF}{{\bf F}}
\newcommand{\bT}{{\bf T}}

\newcommand{\bc}{{\bf c}}

\newcommand{\bnu}{{\bf \nu}}

\newcommand{\cA}{{\cal A}}
\newcommand{\cB}{{\cal B}}
\newcommand{\cC}{{\cal C}}
\newcommand{\cE}{{\cal E}}

\newcommand{\cH}{{\cal H}}

\newcommand{\cL}{{\cal L}}

\newcommand{\cN}{{\cal N}}
\newcommand{\cO}{{\cal O}}
\newcommand{\cR}{{\cal R}}
\newcommand{\cS}{{\cal S}}

\newcommand{\cZ}{{\cal Z}}

\def\ent{{ecoNORTA}}
\newcommand{\dpz}{{\mathbf{z}^*}} 
\newcommand{\dpzsub}[1]{{\mathbf{z}^*_{#1}}} 
\newcommand{\tdpz}{{\tilde{\mathbf{z}}^*}} 
\newcommand{\rstr}{{r^*}} 
\newcommand{\dpzs}{{\cZ^*}} 
\newcommand{\dpzsn}[1]{{\hat{\cZ}^*_{#1}}} 
\newcommand{\Sx}{{\cal S}_\mathbf{x}}
\newcommand{\Sz}{{\cS_{z_1^*,{\mathbf{z}}}}}
\newcommand{\Szcom}[1]{{\cS_{#1,\mathbf{z}}}}

\newcommand{\tSz}{{\cS_{\tilde{\mathbf{z}}}}}
\newcommand{\Szc}[1]{{\cS_{\mathbf{z},#1}}}
\newcommand{\ind}[1]{\mathbb{I}{\left\{#1\right\}}}
\newcommand{\var}{\mathbb{V}\mbox{ar}}
\newcommand{\estAR}{\hat\alpha_{\cA\cR}}
\newcommand{\estNml}{\hat\alpha_{\cN}}
\newcommand{\estIS}{\hat\alpha_{\rm\tiny IS}}

\newcommand{\estExp}{\hat\alpha_{\cE}}
\newcommand{\estLapsub}[1]{\hat\alpha_{\cL,#1}}
\newcommand{\estLap}{\hat\alpha_{\cL}}
\newcommand{\estOpt}{\hat\alpha_{\cO}}
\newcommand{\estNt}{\hat\alpha_{\rm\tiny eN}}

\newcommand{\tndi}{\rightarrow \infty}

\newcommand{\E}{\mathbb{E}}
\renewcommand{\P}{\mathbb{P}}

\newcommand{\mse}{\mbox{MSE}}

\newtheorem{theorem}{Theorem}
\newtheorem{definition}{Definition}
\newtheorem{assumption}{Assumption}
\newtheorem{lemma}{Lemma}
\newtheorem{example}{Example}
\newtheorem{remark}{Remark}

\makeatletter
\newcommand*{\pmzeroslash}{%
  \nfss@text{%
    \sbox0{0}%
    \sbox2{/}%
    \sbox4{%
      \raise\dimexpr((\ht0-\dp0)-(\ht2-\dp2))/2\relax\copy2 %
    }%
    \ooalign{%
      \hfill\copy4 \hfill\cr
      \hfill0\hfill\cr
    }%
    \vphantom{0\copy4 }
  }%
}
\makeatother

\newcommand{\refthm}[1] {{Theorem~\ref{#1}}}
\newcommand{\reffig}[1] {{Figure~\ref{#1}}}
\newcommand{\reflem}[1] {{Lemma~\ref{#1}}}

\newcommand{\refsec}[1] {{Section~\ref{#1}}}
\newcommand{\refasm}[1] {{Assumption~\ref{#1}}}

\newcommand{\ProofEnd}{\mbox{$\Box$} \vspace{ 8pt}}
\newcommand{\ProofOf}[1] {{\noindent \bf Proof of~#1: }}

\algdef{SE}[DOWHILE]{Do}{doWhile}{\algorithmicdo}[1]{\algorithmicwhile\ #1}%
\algdef{SE}[SUBALG]{Indent}{EndIndent}{}{\algorithmicend\ }%
\algtext*{Indent}
\algtext*{EndIndent}

\usepackage{float}
\floatstyle{plaintop}
\newfloat{newtable}{thp}{tablectr}
\floatname{newtable}{Table}

\usepackage{picinpar}

\renewcommand\Authands{ and }
\begin{document}

\title{Efficient Estimation in the Tails of Gaussian Copulas}

\title{Efficient Estimation in the Tails of Gaussian Copulas}
\author[1]{Kalyani Nagaraj \thanks{kalyanin@purue.edu}}
\author[2]{Jie Xu \thanks{jxu13@gmu.edu}}
\author[1]{Raghu Pasupathy\thanks{pasupath@purdue.edu}}
\author[3]{Soumyadip Ghosh\thanks{ghoshs@us.ibm.com}}
\affil[1]{Department of Statistics, Purdue University}
\affil[2]{Department of SEOR, George Mason University}
\affil[3]{Math Sciences and Analytics, IBM Research}

\renewcommand\Authands{ and }

\maketitle

\begin{abstract}
We consider the question of efficient estimation in the tails of Gaussian copulas. Our special focus is estimating expectations over multi-dimensional constrained sets that have a small implied measure under the Gaussian copula. We propose three estimators, all of which rely on a simple idea: identify certain \emph{dominating} point(s) of the feasible set, and appropriately shift and scale an exponential distribution for subsequent use within an importance sampling measure. As we show, the efficiency of such estimators depends crucially on the local structure of the feasible set around the dominating points. The first of our proposed estimators $\estOpt$ is the ``full-information" estimator that actively exploits such local structure to achieve bounded relative error in Gaussian settings. The second and third estimators $\estExp$, $\estLap$ are ``partial-information" estimators, for use when complete information about the constraint set is not available; they do not exhibit bounded relative error but are shown to achieve polynomial efficiency. We provide sharp asymptotics for all three estimators. For the NORTA setting where no ready information about the dominating points or the feasible set structure is assumed, we construct a multinomial mixture of the partial-information estimator $\estLap$ resulting in a fourth estimator $\estNt$ with polynomial efficiency, and implementable through the ecoNORTA algorithm. Numerical results on various example problems are remarkable, and  consistent with theory.
\end{abstract}




%

\section{INTRODUCTION}\label{sec:intro}

We investigate the question of efficiently estimating nonlinear expectations 
on constrained sets, that is, quantities that can be expressed as
\begin{equation}\label{alphaexp}
\alpha \defn \E[h(\bX) \,\,\ind{ \bX\in\Sx}],
\end{equation}
where $h(\cdot)$ is a known polynomial, $\Sx$ is a known constraint set, and $\bX$ has the
NORTA (Gaussian copula) distribution~\citep{nel2013,mcnfreemb2005,nelsen1}. An important
special case is the context of estimating the probability $\alpha \defn \P(\bX \in \Sx)$ assigned to the set $\Sx$ by a Gaussian copula, obtained by setting $h(\cdot) \equiv 1$ in (\ref{alphaexp}). Since $\bX$ belongs to the NORTA family, we assume that random variates from the specific distribution of $\bX$ can be generated rapidly~\citep{cario1,cario2} on a digital computer. Also, when we say the function $h$ and the set $\Sx$ are ``known," we mean that they are both expressed in analytical form and that their structural properties, such as the curvature at a given point, can be deduced with some effort. Our particular interest is an estimator for $\alpha$ that is effective when the set $\Sx$ is assigned a small measure, as might happen in the context of studying the occurrence of rare events and calculating associated expectations in physical systems modeled using Monte Carlo simulation.

The question of estimating an expectation on a ``small" constrained set is well motivated~\citep{2011krotaibot,asmgly07}, with examples arising in diverse fields such as production systems~\citep{1996glaliu}, epidemics modeling~\citep{eubguckummarsritorwan2004,dimmey2010}, reliability settings~\citep{1987barpro,2009lee}, financial applications~\citep{mcnfreemb2005,gla04}, and confidence set construction within statistics~\citep{das2008,das2011}. The problem setting we consider in this paper is specific in that $\bX$ in (\ref{alphaexp}) is assumed to be a NORTA random vector. We believe this special case is worthy of investigation since NORTA random vectors have recently become an important modeling paradigm~\citep[Chapter 5]{mcnfreemb2005} and, as we shall see, the knowledge that $\bX$ has a Gaussian copula can lead to highly efficient estimators of $\alpha$. Efficiency, as is usual, is considered here in a certain asymptotic sense, as the measure assigned to the set $\Sx$ tends to zero.

\subsection{Two Natural Estimators}\label{sec:natural}
An obvious consideration for estimating $\alpha$ in (\ref{alphaexp}) is the \emph{acceptance-rejection} estimator, where independent and identically distributed (iid) copies 
of $\bX$ are generated and an estimator $\estAR$ of $\alpha$ is constructed using those random variates that fall within the set $\Sx$. To see why this estimator may not be efficient, consider estimating $\alpha(z_1^*) \defn \P(\bX > z_1^*)$ obtained by setting $\Sx \defn (z_1^*,\infty)$ and $h(\cdot) \equiv 1$ in (\ref{alphaexp}). (In what follows, we treat the quantity $\alpha$ in (\ref{alphaexp}) as a function of the parameter $z_1^*$ for reasons that will become clear.) The acceptance-rejection estimator $\estAR(z_1^*)$ is then given by $\estAR(z_1^*) \defn \ind{\bX > z_1^*}$. With some algebra, one can show that $\E[\estAR(z_1^*)] = \alpha(z_1^*)$ and that $\var(\estAR(z_1^*))= \alpha(z_1^*)(1-\alpha(z_1^*))$, giving the \emph{relative error} \begin{equation}\label{naivecv}\mbox{RE}(\estAR(z_1^*)) = \frac{\sqrt{\var(\estAR)}}{\E[\estAR]} = \sqrt{\frac{1-\alpha(z_1^*)}{\alpha(z_1^*)}}.\end{equation} We see from (\ref{naivecv}) that the relative error $\mbox{RE}(\estAR(z_1^*)) \to \infty$ as $z_1^* \to \infty$, and particularly that $\mbox{RE}(\estAR(z_1^*)) \sim \sqrt{\alpha^{-1}(z_1^*)}$. (For two positive sequences $\{a_n\}, \{b_n\}$ converging to zero, we say $a_n \sim b_n$ to mean that $\lim_{n \to \infty} a_n/b_n = 1$.) Moreover, if $\bX$ has the Gaussian distribution with zero mean and unit variance, $\alpha^{-1}(z_1^*) \sim z_1^*\sqrt{2\pi}\exp\{{z_1^*}^2/2\} \to \infty$ at an exponential rate (as $z_1^* \to \infty$) suggesting that $\estAR(z_1^*)$ is a poor estimator of $\alpha(z_1^*)$, especially for large values of $z_1^*$.  


A more sophisticated way of estimating $\alpha$ is through exponential tilting or twisting~\citep{gla04}, where an estimator $\estNml$ is obtained through importance
sampling with a ``shifted joint-normal" followed by an acceptance-rejection step. For the example considered above where $\alpha(z_1^*) \defn \P(\bX > z_1^*)$, the exponential-twisting estimator $\estNml$ 
\begin{equation}\label{eq:alphahatnml}
  \estNml(z_1^*)= h(\tilde{X}) \,\ind{\tilde{X} \in\Sx}\,\,\left(
  \frac {\phi(\tilde{X})} {\phi(\tilde{X};{\mu},\sigma^2)}\right) = \ind{\tilde{X} \in (z_1^*,\infty) }\,\,\left(
  \frac {\phi(\tilde{X})} {\phi(\tilde{X};{\mu},\sigma^2)}\right),
\end{equation}
where $\tilde{X}$ has Gaussian density $\phi(\cdot;\mu,\sigma^2)$ with mean $\mu$ and variance $\sigma^2$, and $\phi(\cdot)$ is the standard Gaussian density having mean zero and unit variance.  

The estimator $\estNml$, like the estimator $\estAR$, is unbiased with respect to $\alpha$. Theorem~\ref{thm:nonexistence} formally characterizes the asymptotic variances of $\estAR$ and $\estNml$ through a relative error calculation. A proof of Theorem~\ref{thm:nonexistence} can be found in Appendix \ref{techproofs}.
\begin{theorem}\label{thm:nonexistence}
Let $Z$ be the standard Gaussian random variable, $\Sx \defn (z_1^*,\infty)$, and $\alpha = \E[\mathbb{I}\{Z \in \Sx\}] = \P(Z \in \Sx)$. Then the following hold as $z_1^* \to \infty$.
  \begin{itemize}
\item[(a)] $\alpha^{-2}(z_1^*)\E\left[\estAR^2(z_1^*)\right] \,\, \sim \,\,
  \sqrt{2\pi} {z_1^*} \exp\{\frac{1}{2}{z_1^*}^2\}.$
\item[(b)] $\alpha^{-2}(z_1^*)\E\left[\estNml^2(z_1^*)\right] \,\, \sim \,\,
  {z_1^*}^2 \exp\{\frac{1}{2}(\mu - z_1^*)^2\},$ with the minimum squared relative error $$\inf_{\mu} \alpha^{-2}(z_1^*)\E\left[\estNml^2(z_1^*)\right] \,\, \sim \,\,
  {z_1^*}^2$$ attained for the choice $\mu ={z_1^*} + o({z_1^*})$.
  \end{itemize}
\end{theorem}

It is evident from Theorem \ref{thm:nonexistence} that if $\mu$ is chosen carefully, the estimator $\estNml$ satisfies $\mbox{RE}(\estNml(z_1^*)) \sim z_1^*$. While this suggests that $\estNml$ is a much better estimator than $\estAR$, the fact remains that $\mbox{RE}(\estNml(z_1^*))$ goes to $\infty$ as $z_1^* \to \infty$ even in the one-dimensional context. By contrast, the estimator $\estOpt$ that we propose enjoys $\mbox{RE}(\estOpt(z_1^*)) \to 1$ as $z_1^* \to \infty$; in fact, we show that $\estOpt$ achieves bounded relative error for more general problems in an arbitrary (but finite) number of dimensions under certain conditions.   

\subsection{Summary and Key Insight}

The central question underlying our investigation is whether there exist (Monte Carlo) estimators of $\alpha$ whose relative error remains bounded as the set $\Sx$ becomes rare in a certain sense. We answer in the affirmative but with some qualifications. We argue that highly efficient estimators of $\alpha$, particularly those with bounded relative error, can be constructed through the use of an appropriately shifted and scaled exponential importance sampling measure. The extent of such shifting and scaling, however, depends crucially on the following three structural properties of the set $\Sx$: (i) location of certain \emph{dominating points} in $\Sx$ defined (loosely) as the set of points that contribute maximally to the calculation of $\alpha$; (ii) the local curvature of the set $\Sx$ at the dominating points; and (iii) the existence (or lack) of a supporting hyperplane to the set $\Sx$ at the dominating points. Considering (i) -- (iii), we propose three alternate estimators $\estOpt$, $\estExp$, $\estLap$ that become applicable depending on the setting. The applicability of the estimators $\estOpt$, $\estExp$, $\estLap$ is summarized in Table \ref{fig:portfolio} and illustrated through Figure \ref{fig:structknowledge}.

\begin{figure*}[!h]
   \centering
   \includegraphics[trim=5cm 9cm 3cm 8cm, clip=true, width= \textwidth,angle=0]{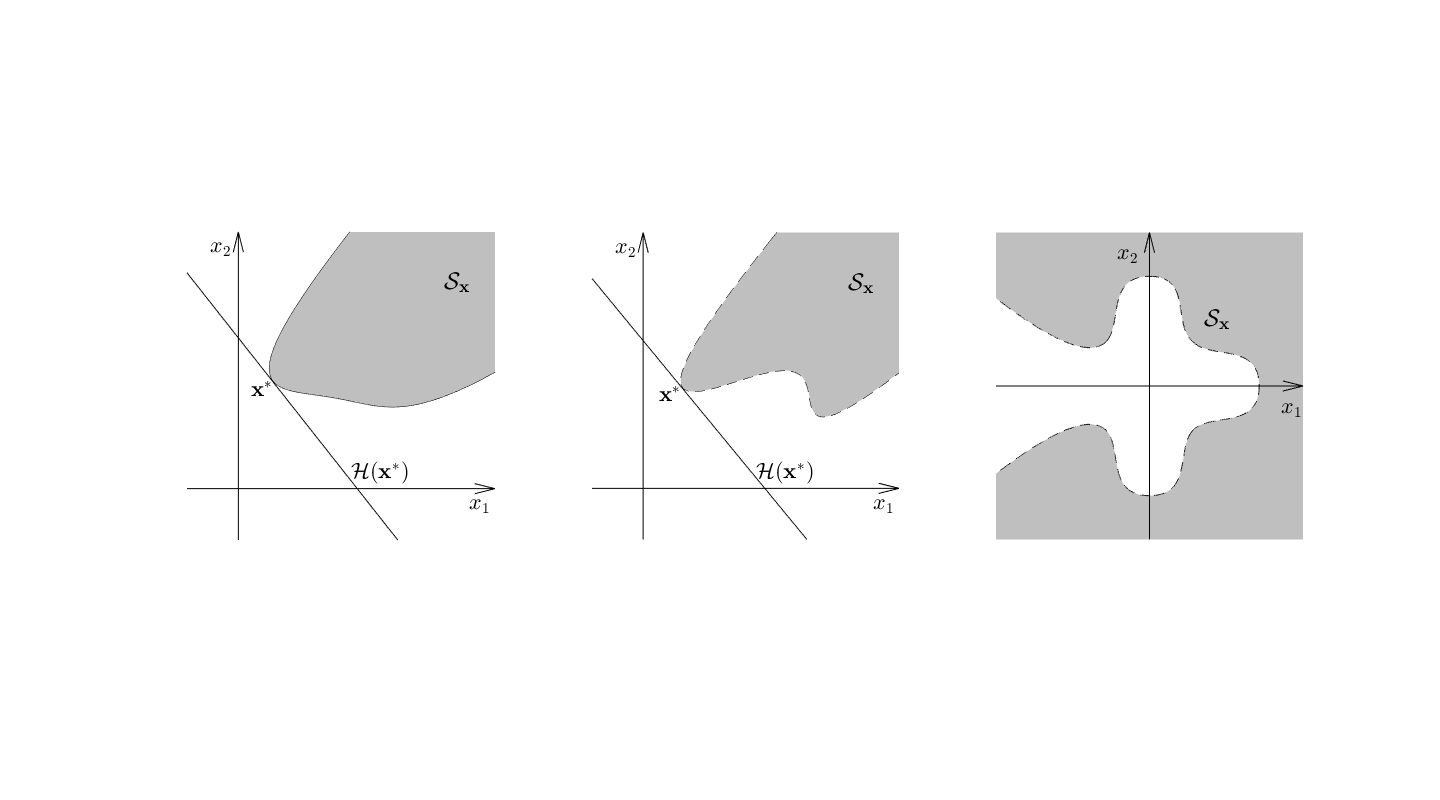}
   \caption{The figure illustrates three settings where the proposed estimators $\estOpt, \estExp, \estLap$ are applicable, respectively. The leftmost panel depicts a setting where the unique \emph{dominating point} $x^*$ of $\cS_{\bX}$, the local structure of $\cS_{\bX}$ around $x^*$, and a supporting hyperplane $\cH$ at $x^*$ are known, thereby allowing the use of the full-information estimator $\estOpt$. The center panel depicts a setting where the boundary (and the curvature) of the set $\cS_{\bX}$ is unknown but a dominating point $x^*$ and a hypeplane $\cH$ at the dominating point $x^*$ are known. The rightmost panel illustrates a setting with multiple unknown dominating points and where there exists no supporting hyperplane to the set $\cS_{\bX}$.}
   \label{fig:structknowledge}
\end{figure*} 

Which amongst the estimators $\estOpt$, $\estExp$, $\estLap$ is most appropriate for a given setting will be dictated by how much information we have about (i), (ii), and (iii). For instance, the estimator $\estOpt$ is the ``full information" estimator designed for use in contexts where $\bX$ is Gaussian, and the set $\Sx$ is well-behaved with known structural properties, that is, $\Sx$ has a unique dominating point with identifiable local curvature that is encoded by certain structural constants, and has a supporting hyperplane at the dominating point. We argue later that the conditions on $\Sx$ under which the full-information estimator $\estOpt$ becomes applicable are not onerous. Particularly, we show in Section \ref{sec:etaands} that for a large class of sets $\Sx$, the structural constants of $\Sx$ are identifiable through a linear program that is easily solved. We provide sharp asymptotics of $\estOpt$ in Section \ref{sec:mvneff}, demonstrating that it enjoys bounded relative error leading to the remarkable numerical performance illustrated in Section \ref{sec:mvnexpts}. 

\begin{table}[h!]
   \centering
   \includegraphics[trim=2cm 11cm 2cm 11cm, clip=true,width= \textwidth,angle=0]{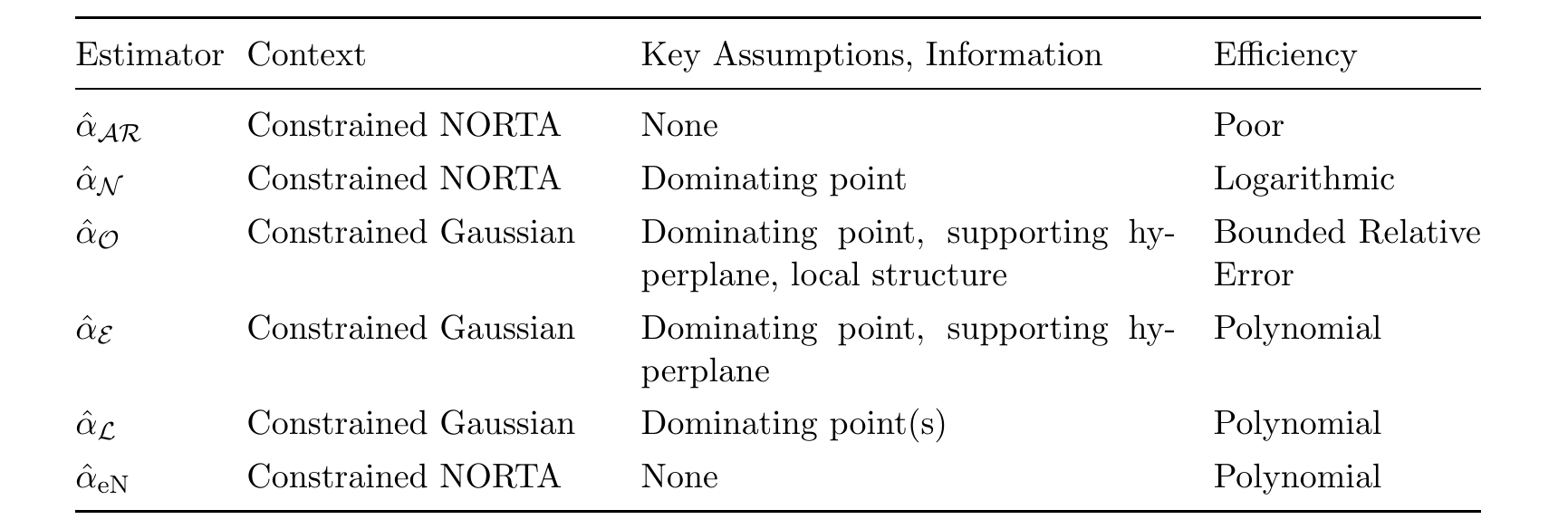}
   \caption{We propose the three estimators $\estOpt$, $\estExp$, and $\estLap$ for use depending on available structural information. The estimator $\estOpt$ achieves bounded relative error, but relies on the explicit use of local structure of the feasible set. The estimators $\estExp$, and $\estLap$ are more general but exhibit polynomial efficiency.}
   \label{fig:portfolio}
\end{table} 

The other two estimators $\estExp$ and $\estLap$ that we propose are ``partial-information" estimators in that they are applicable in Gaussian settings where we have knowledge of the dominating point of $\Sx$ but only limited to no curvature information. As we show in Section \ref{sec:mvnpartial} where we provide the sharp asymptotics for  $\estExp$ and $\estLap$, such lack of full information hinders the optimal choice of importance-sampling parameters resulting in a loss of the bounded relative error property. The estimators $\estExp$ and $\estLap$ still achieve a weaker form of efficiency that we call polynomial efficiency.

For contexts where no information about (i) -- (iii) is available, we first propose and analyze an (unimplementable) estimator $\estNt$ that is obtained as a multinomial mixture of the partial-information estimator $\estLap$. The ecoNORTA algorithm that we propose then constructs a sequential form of $\estNt$ by adaptively estimating the dominating points. ecoNORTA starts with an initial crude guess of the dominating points, and as random variates are generated during the estimation process, progressively updates the location of the dominating points within the estimator $\estLap$. 

\subsection{Paper Organization}

The paper is organized into two main parts that appear in Section \ref{sec:mvn} and Section \ref{sec:norta} respectively. Section \ref{sec:mvn} treats the Gaussian context in its entirety; we present the full-information estimator $\estOpt$ in Section \ref{sec:mvneff} and the partial-information estimators $\estExp, \estLap$ in Section \ref{sec:mvnpartial}. Much of the theoretical machinery introduced in Section \ref{sec:mvn} is then co-opted into Section \ref{sec:norta} where we treat the NORTA context. Section \ref{numerics} provides numerical illustrations in the Gaussian and the NORTA contexts. In the ensuing section, we first introduce some important notions of asymptotic efficiency that will be used throughout the paper.

\section{ASYMPTOTIC EFFICIENCY: NOTATION AND DEFINITIONS}\label{sec:notation}

As is usual in rare-event literature, we use the notion of \emph{relative error} in assessing the efficiency of the estimators of $\alpha$. The relative error $\mbox{RE}(\hat{\alpha})$ of the estimator $\hat{\alpha}$ (with respect to the quantity $\alpha$) is given by \begin{equation}\label{relerror}\mbox{RE}(\hat{\alpha}) \defn \sqrt{\frac{\mbox{MSE}(\hat{\alpha},\alpha)}{\alpha^2}}; \,\, \mbox{MSE}(\hat{\alpha},\alpha) \defn \E[(\hat{\alpha} - \alpha)^2].\end{equation} Since much of our analyses is in a ``rare-event regime," we will assess $\hat{\alpha}$ by studying the behavior of $\mbox{RE}(\hat{\alpha})$ as $\alpha \to 0$, where the latter limit is usually accomplished by sending a ``rarity parameter" $z_1^* \to \infty$. Accordingly, we will be compelled to use the notation $\alpha(z_1^*)$ and $\hat{\alpha}(z_1^*)$ to make explicit the dependence of $\alpha$ and $\hat{\alpha}$ on the rarity parameter $z_1^*$.

The following notions that quantify the asymptotic behavior of the relative error will be useful for assessing estimators.
\begin{definition}{(Bounded Relative Error)}\label{defn:bre}
The estimator $\hat{\alpha}(z_1^*)$ is said to exhibit bounded relative error (BRE)
if $\limsup_{z_1^* \to \infty} \, \mbox{RE}(\hat{\alpha}(z_1^*)) = b_{\hat{\alpha}} < \infty.$ \end{definition} That an estimator has BRE means that its root mean squared error tends to zero at a rate that is commensurate with the rate at which the quantity $\alpha$ it estimates tends to zero. Estimators exhibiting BRE are generally difficult to find in the Monte Carlo context; those with $b_{\hat{\alpha}} = 0$ are especially difficult to find.   

Considering the difficulty of finding estimators exhibiting BRE, a weaker form of efficiency called \emph{logarithmic efficiency} has become popular. \begin{definition}{(Logarithmic Efficiency)}\label{defn:logeff} The estimator $\hat{\alpha}(z_1^*)$ is said to exhibit logarithmic efficiency
if $\limsup_{z_1^* \to \infty} \alpha^{\epsilon - 2}(z_1^*)\,\mse(\hat{\alpha}(z_1^*)) = 0 \quad  \mbox{for all } \epsilon > 0.$ Equivalently, logarithmic efficiency holds if $$\liminf_{z_1^* \to \infty} \, (\log \alpha^2(z_1^*))^{-1}\log \mse(\hat{\alpha}(z_1^*)) \geq 1.$$
\end{definition} In this paper, we use a slightly more specific form of efficiency called \emph{polynomial efficiency} to characterize the behavior of estimators that do not exhibit BRE. \begin{definition} An estimator $\hat{\alpha}(z_1^*)$ is said to exhibit Polynomial($s$), $s \in (0,\infty)$ efficiency if
$$\limsup_{z_1^* \to \infty} \, {z_1^*}^{-s}\,\mbox{RE}(\hat{\alpha}(z_1^*)) < \infty.$$ \end{definition}
It is clear that if $\hat{\alpha}(z_1^*)$ is Polynomial$(s)$ efficient, then it is Polynomial$(s')$ efficient for all $s' \geq s$. Hence, we will generally seek the smallest $s$ such that a given estimator $\hat{\alpha}(z_1^*)$ is Polynomial($s$) efficient. Also, it can be shown with some algebra that estimators that exhibit Polynomial($s$) efficiency are also logarithmically efficient as long as $\alpha(z_1^*)$ converges to zero ``faster" than ${z_1^*}^{-s}$, that is, ${z_1^*}^s \alpha^{\epsilon}(z_1^*) \to 0$ as $z_1^* \to \infty$ for any $\epsilon > 0$. See~\cite{asmgly07} for more on measures of efficiency in the rare event simulation context.

\section{THE CONSTRAINED GAUSSIAN CONTEXT}\label{sec:mvn}

In this section, we treat the special context of estimation on low-probability sets driven by the Gaussian measure, that is, the question of estimating $\alpha = \E[h(\bX) \,\,\ind{ \bX \in \Sx}]$ when $\bX$ has a Gaussian distribution. As we shall see, the constrained Gaussian context is special in that knowledge of the local structure of the set $\Sx$ at the so-called dominating point can be used fruitfully in constructing highly efficient estimators of $\alpha$. Accordingly, in this section, we propose and analyze three different estimators $\estOpt, \estExp, \estLap$ depending on the extent of such available information. We first reformulate the problem statment for ease of exposition.

\subsection{Problem Reformulation}
For clarity, and since we are in the constrained Gaussian setting, we specialize the notation introduced earlier to write 
\begin{align}\label{gauss_pb}
\alpha=\E[h(\bZ) \,\,\ind{ \bZ\in\Sz}],
\end{align} where the feasible set $\Sz \subset \real^d$, and $h$ is a polynomial in $\bz$. The first subscript $z_1^* \in \real$ refers to the ``rarity parameter" and will be explained in greater detail in Section \ref{sec:asympregimes}. It has been introduced into notation to make explicit the dependence of the feasible set $\Sz$ on a parameter $z_1^*$ that will be sent to infinity in our asymptotic analyses.

\subsubsection{Key Assumptions} For the purposes of Section \ref{sec:mvn} alone, we assume that the random vector $\bZ$ and $\Sz$ in (\ref{gauss_pb}) are expressed in such a way that the following assumption holds. 

\begin{assumption}\label{ass:reform} \begin{enumerate} \item[\rm{(a)}] $\bZ$ is distributed as the standard Gaussian density $$\phi(\bz) = ({2\pi})^{-d/2}\exp\{-\frac{1}{2}\bz^T\bz\}, \,\, \bz \in \real^d.$$ \item[\rm{(b)}] The ``dominating point" $\mathbf{z}^* \defn \arg\inf_{z \in \Sz} \{\|\mathbf{z}\|\}$ exists, is unique, and known. \item[\rm{(c)}] The ``dominating point" $\bf{z}^*$ is such that $z_2^* = z_3^* = \cdots = z_d^* = 0$. \end{enumerate}\end{assumption}

Assumption\ref{ass:reform}(a) does not threaten generality --- settings where $\bX$ has a multivariate Gaussian distribution with mean $\mathbf{\mu} \in \real^d$ and positive-definite covariance matrix $\Sigma$ can be recast in the ``standard Gaussian space" as the problem of estimating $\alpha=\E[h(\bZ)\ind{ \bZ\in A^{-1}(\Sx - \mathbf{\mu})}]$, where the $d \times d$ lower-triangular matrix $A$ is such that $AA^T = \Sigma$, and the set $A^{-1}(\Sx - \mathbf{\mu}) \defn \{\mathbf{z}: \mathbf{z} = A^{-1}(\mathbf{x} - \mathbf{\mu}), \mathbf{x} \in \Sx\}$.

Assumption\ref{ass:reform}(b) holds often. For example, when the set $\Sz$ is expressible through known convex constraints, that is, $\Sz \defn \{\mathbf{z}: \ell_i(\mathbf{z}) \geq 0, i=1,2, \ldots,m\}$ where the functions $\ell_i(\mathbf{z}): \real^q \to \real, i=1,2,\ldots,m$ are convex, then $\dpz \defn \arg\min_{\mathbf{z}\in\Sz}\|\mathbf{z}\|_2$ is the solution to a convex optimization problem and can usually be identified simply. Even if one or more of the functions $\ell_i, i=1,2,\ldots,m$ are not convex, $\dpz$ could probably be identified, albeit with some effort, by solving a non-convex optimization problem. If $\Sz$ is not expressed through the constraint functions but is instead expressed through a membership oracle, identifying $\dpz$ could become a challenging proposition.

Assumption\ref{ass:reform}(c), which stipulates that the second through the $d$th coordinates of $\bf{z}^*$ are zero, has been imposed for convenience and can be ensured through an appropriate rotation of the set $\Sz$. Specifically, suppose we wish to estimate $\alpha=\E[h(\bZ)\ind{ \bZ \in \Sz}]$ where $\bZ$ is a standard Gaussian random vector and suppose $\dpz \defn \arg\inf_{z \in \mathbf{\Sz}} \{\|\mathbf{z}\|\}$ exists and is unique. Then, since the standard Gaussian distribution is spherically symmetric, we can transform the problem using an appropriate ``rotation matrix" $R(\dpz)$ calculated such that $\tdpz = (\tilde{z}_1^*, \tilde{z}_2^*, \ldots, \tilde{z}_d^*)  \defn R(\dpz)\dpz$ satisfies $\tilde{z}_2^* = \tilde{z}_3^* = \cdots = \tilde{z}_d^* = 0$. Such a rotation matrix always exists. The corresponding constraint set after such rotation becomes $\tSz \defn \{R(\dpz)\bz: \bz \in \Sz\}$ yielding the reformulated problem of needing to estimate $\alpha=\E[h(R^{-1}(\dpz)\tbZ)\ind{ \tbZ \in \tSz}]$. Furthermore, since the function $h$ is a polynomial, $h \circ R^{-1}(\bz)$ is also a polynomial and no generality is lost.  

Considering the above discussion, reformulating the general Gaussian setting to satisfy Assumption\ref{ass:reform}(b) and Assumption\ref{ass:reform}(c) involves two steps in succession: (i) standardize the set $\Sx$ to $A^{-1}(\Sx - \mathbf{\mu})$, and (ii) perform the rotation $\tbz \defn R(\dpz)\bz$.


We emphasize that Assumption\ref{ass:reform} is a standing assumption in the Gaussian context, that is, all three estimators $\estOpt$, $\estExp$, $\estLap$ rely on it. By contrast, the following structural assumption is needed for constructing only $\estOpt$ and $\estExp$, and not $\estLap$. 

\begin{assumption}\label{supphyp} The hyperplane ${\mathbf{z^*}}^T\mathbf{z} = \|\mathbf{z^*}\|^2$ is a supporting hyperplane to the set $\Sz$, that is, every point $\mathbf{z}_0 \in \Sz$ satisfies ${\mathbf{z^*}}^T\mathbf{z}_0 \geq \|\mathbf{z^*}\|^2$.
\end{assumption}

Assumption\ref{supphyp} states that the hyperplane passing through the dominating point $\dpz$ and normal to the line joining the origin and the point $\dpz$ is such that the set $\Sz$ is on one side of it. The spirit of Assumption\ref{supphyp} is that the region of integration governing the calculation of $\alpha$ is a ``tail region" that is a subset of an appropriate half-space. To aid reader's intuition, we note that Assumption\ref{ass:reform} and Assumption\ref{supphyp} together imply that the sets $\Sz$ that we consider in Gaussian context have a dominating point and a ``vertical" supporting hyperplane $z_1 = z_1^*$ to $\Sz$ at $(z_1^*,0,0,\ldots,0)$.  

\subsubsection{Asymptotic Regimes and a Word on Notation}\label{sec:asympregimes} We will consider two types of asymptotic regimes when analyzing the effectiveness of a proposed estimator of $\alpha$. The first of these regimes, called the \emph{translation regime}, refers to the sequence (in $z_1^*$) of sets $\Sz$ obtained by ``translating" a fixed set along the $z_1$-axis. Formally, for a fixed set $\Szcom{0}$, we obtain the translation regime by defining $\Sz \defn \{(z_1,z_2,\ldots,z_d) : (z_1 - z_1^*,z_2,\ldots,z_d) \in \Szcom{0}\}$ and then considering the sequence of sets $\Sz$ as $z_1^* \to \infty$. 

The second asymptotic regime, called the \emph{scaling regime}, refers to the sequence of sets $\Sz$ obtained by ``scaling" points in some fixed set $\Szcom{0}$ using the scalar $z_1^*$. Formally, for a fixed set $\Szcom{0}$, we obtain the scaling regime by defining $\Sz \defn \{ (z_1,z_2,\ldots,z_d) : (z_1/z_1^*, z_2/z_1^*, \ldots,z_d/z_1^*) \in \Szcom{0}\}$ and then considering the sequence of sets $\Sz$ as $z_1^* \to \infty$. 

Considering the reformulation due to Assumption\ref{ass:reform}, under each
regime the probability assigned to $\Sz$ vanishes by simply sending $z_1^* \to \infty$.  The translation and scaling regimes are depicted in Figure \ref{fig:asymregime} where a fixed set $\Szcom{0}$ is either translated or scaled to obtain the needed asymptotic regime.
\begin{figure}\label{fig:asymregime}
\centering
\includegraphics[height=2.5in]{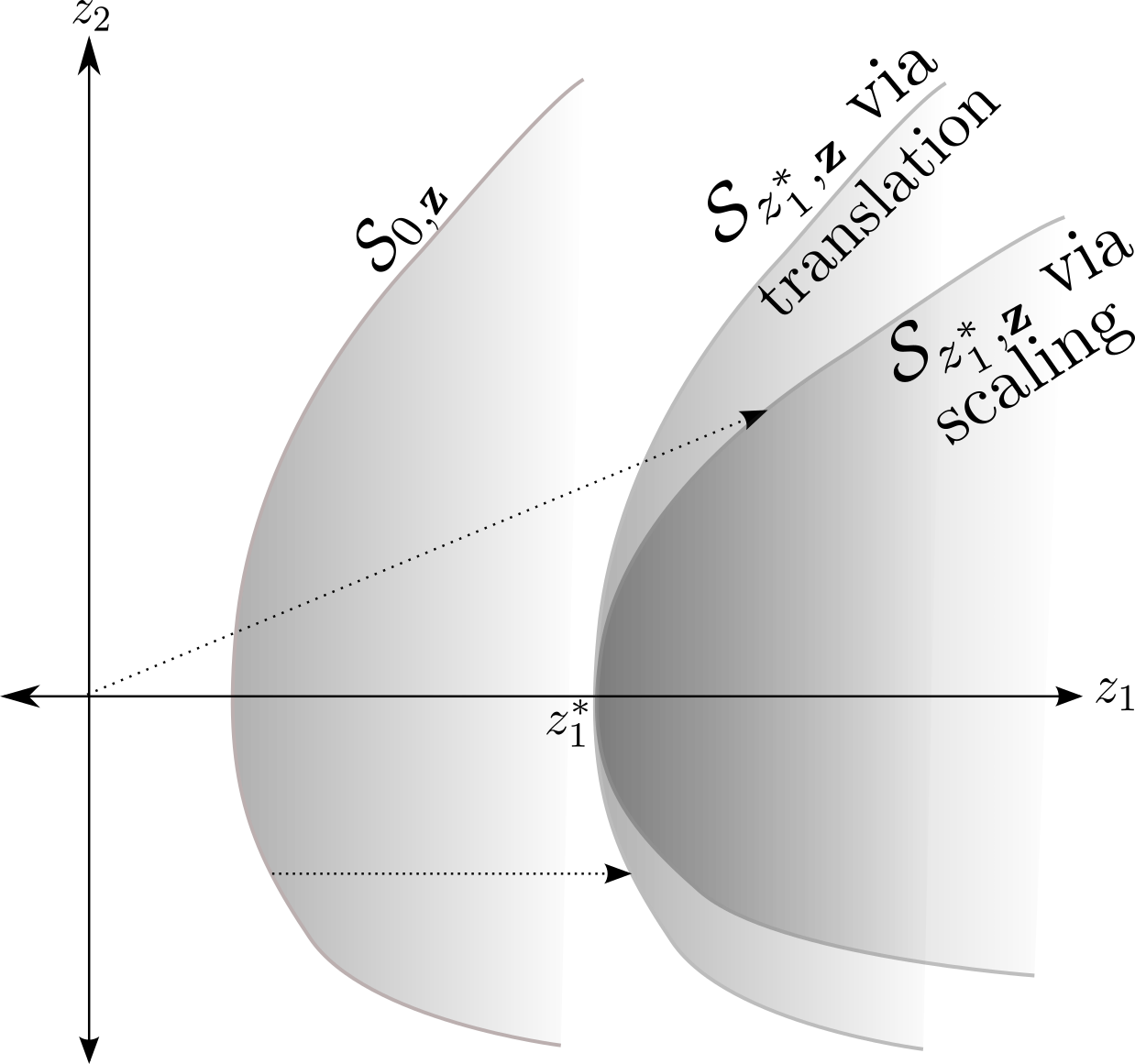}
\caption{We consider the translation and the scaling asymptotic regimes in characterizing the efficiency of the estimators that we propose.}
\end{figure}

The following comments are aimed at further clarifying notational issues.

\begin{enumerate} \item[(i)] The fixed set $\Szcom{0}$ in the discussion above was introduced expressly for explaining the translation and scaling asymptotic regimes. We find no reason to refer to the set $\Szcom{0}$ anywhere in the rest of the paper. 

\item[(ii)] Throughout the paper the scalar $z_1^*$ will serve as the rarity parameter that will be sent to infinity. In the translation regime, $z_1^*$ turns out to be the first-coordinate of the unique dominating point. In the scaling regime, $z_1^*$ has no such physical meaning.

\item[(iii)] Unless mentioned explicitly, all analysis of the estimators we propose are performed in the translation regime. Particularly, all analysis in Section \ref{sec:mvneff} and Section \ref{sec:mvnpartial} is in the translation regime. We are forced to undertake analyses in the scaling regime in Section \ref{sec:norta} in order to contend with the possibility of multiple dominating points.

\end{enumerate}

\subsection{The Full-Information Estimator $\estOpt$} \label{sec:mvneff}

We now propose the full-information estimator $\estOpt$ for the constrained Gaussian context and in general high-dimensional space. As we shall see, knowledge of the local structure of the set $\Sz$ at the dominating point $\dpz \defn \arg\inf_{\bz \in \mathbf{\Sz}} \{\|\mathbf{z}\|\}$ is crucial for constructing efficient estimators of $\alpha$. Accordingly, our proposed estimator $\estOpt$ is a function of the local structure of the set $\Sz$ around the point $\dpz$. Such local structure is encoded through the function $v_p(t)$ in the following assumption, where $v_p(t)$ quantifies the ``cost-scaled content" of the set $\Sz$ close to the point $\dpz$ and ``about" the line joining $\dpz$ to the origin. When the cost function $h(\cdot)$ is identity, for instance, the function $v_p(t)$ simply connotes the volume of the ``cross-section" of the set $\Sz$ when $z_1^* = t$.

\begin{assumption}\label{curvcoeffknown} Let the cost function $h(\bz)$ be a polynomial in $\bz \defn (z_1,z_2,\ldots,z_d)$, conveniently expressed as $h(\bz) = \sum_{i=1}^{p} \gamma_i z_1^i\prod_{j=2}^d z_j^{c_{j}(i)}$. Thus, $\gamma_{p}z_1^{p}\prod_{j=2}^d z_j^{c_{j}(p)}$ is the term corresponding to the largest power of $z_1$. Then, denoting $\bc \defn (p, c_2(p),, c_3(p), \ldots,c_d(p))$ and the cross-section set $\Szcom{z_1^*}(t) \defn \{(z_2,z_3,\ldots,z_d): \bz \in \Sz, z_1 - z_1^* = t\}$, the (d-1)-dimensional cost-scaled volume $$v_p(t) = \int_{\Szcom{z_1^*}(t)} \gamma_p \prod_{j=2}^d z_j^{c_{j}(p)}  \, dz_2\,\cdots dz_d$$ satisfies the expansion $v_p(t) = \eta(\bc) t^{s(\bc)} + o(t^{s(\bc)})$ as $t \to 0$, and the constants $\eta(\bc),s(\bc) \in (0,\infty)$ are known. The argument $\bc$ in the constants $\eta(\bc)$ and $s(\bc)$ are often suppressed for convenience.\end{assumption}

The assumption about the existence of a local polynomial expansion of $v_p(t)$ is arguably mild; for example, it only precludes sets $\Sz$ that are not ``too sharp" around the point $\dpz$. The constants $\eta$ and $s$ appearing in the expansion of $v_p(t)$ may or may not be easy to deduce depending on how the set $\Sz$ is expressed. For example, when $\Sz$ is expressed using constraint functions as $\Sz \defn \{\mathbf{z}: \ell_i(\mathbf{z}) \geq 0, i=1,2, \ldots,m\}$ and the functions $\ell_i(\cdot)$ are expressed analytically, the constants $\eta$ and $s$ can be identified easily. Such contexts seem typical and we provide a systematic way of identifying the structural constants $\eta, s$ for a large class of sets $\Sz$ in Section \ref{sec:etaands}.

Under Assusmption\ref{ass:reform}, the quantity of interest $\alpha$ takes the form \begin{align}\label{alphasimple} \alpha = \int_{z_1^*}^{\infty} \phi(z_1) \int_{\Sz(z_1-z_1^*)} h(\bz) \prod_{i=2}^d \phi(z_i),\end{align} where $\eta, s$ are constants appearing in Assumption\ref{curvcoeffknown}, $\phi(x) = (2\pi)^{-1/2}\exp\{-\frac{1}{2}x^2\}$ is the univariate standard Gaussian density, and the cross-section set $\Sz(z_1-z_1^*) \defn \{(z_2,z_3,\ldots,z_d): \bz \in \Sz\} \subset \real^{d-1}$. (In (\ref{alphasimple}) and throughout the paper, we have chosen to omit the tedious repitition of the elemental $(d-1)$-dimensional volume $dz_2 \cdots \,dz_d$ in the integral.) The above form of $\alpha$ inspires our proposed estimator $\estOpt$ which takes the following simple form:
\begin{align}\label{estopt} \estOpt = (\frac{1}{2\pi})^{\frac{d-1}{2}}\frac{\phi(W)}{f_{\lambda(z_1^*)}(W)}\eta {z_1^*}^p 
(W-z_1^*)^s,
\end{align} where $f_{\lambda(z_1^*)}(x) = \lambda(z_1^*)e^{-\lambda(z_1^*)(x-z_1^*)} \mathbb{I}(x \geq z_1^*)$ is the shifted-exponential density function, $W$ is a random variable having the density $f_{\lambda(z_1^*)}(x)$, and the constants $\eta,s$ are from Assumption \ref{curvcoeffknown}. (Theorem \ref{thm:mserate} will establish that the choice $\lambda(z_1^*) = (1+s)^{-1} z_1^*$ is optimal in the sense of minimizing the relative error of $\estOpt$.)

The basis of $\estOpt$ should be evident from the structure of $\alpha$ in (\ref{alphasimple}); the outer integral in (\ref{alphasimple}) is approximated by the term ${z_1^*}^p\eta
(W_{1}-z_1^*)^s$ after noting that $$\int_{\Sz(z_1-z_1^*)} \prod_{i=2}^d \phi(z_i) = (2\pi)^{-(d-1)/2}\left(\eta (z_1 - z_1^*)^s + o\left((z_1 - z_1^*)^s\right)\right)$$ as $z_1 \to z_1^*$ by Assumption\ref{curvcoeffknown}, and the inner integral in (\ref{alphasimple}) is estimated using a shifted-exponential importance sampling measure along the first dimension. There appears to be no strong physical justification for the use of the shifted exponential family but an algebraic explanation is evident by observing the respective exponents of the Gaussian and the exponential densities. The choice of support for the importance sampling measure is dictated by information contained in Assumption~\ref{supphyp}. 

Towards establishing the bounded relative error property of $\estOpt$, we first present a result on the asymptotic expansion of a certain class of integrals that we will repeatedly encounter.

\begin{lemma}\label{lem:assrates} Let $g(t) = \beta_0t^{\beta_1} + o(t^{\beta_1}), \beta_0,\beta_1 \in (0, \infty)$ as $t \to 0$. Then, for $q \in \real$ and as $x^* \to \infty$, the following hold.
\begin{enumerate} \item[\rm{(i)}] $I_1(x^*) \,\, \defn \,\, \int_{x^*}^{\infty} x^q\exp\{-\frac{1}{2}{x}^2\}g(x-x^*)\,dx \,\, \sim \,\, \beta_0\Gamma(\beta_1+1){x^*}^{q-1-\beta_1}\exp\{-\frac{1}{2}{x^*}^2\}.$ \item[\rm{(ii)}] $I_2(x^*) \,\, \defn \,\, \int_{x^*}^{\infty} x^q\exp\{-\kappa x\}g(x-x^*)\,dx \,\, \sim \,\, \beta_0 {x^*}^q\exp\{-\kappa x^*\} \Gamma(\beta_1 + 1)\kappa^{-\beta_1 - 1}.$ \end{enumerate} \end{lemma}

\begin{proof}{Proof.} Notice that \begin{align}\label{I1red} I_1(x^*)  & \defn  \exp\{-\frac{1}{2}{x^*}^2\} \int_{x^*}^{\infty} x^q\exp\{-\frac{1}{2}({x}^2 - {x^*}^2)\} g(x-x^*) \,dx \nonumber \\
& =  \exp\{-\frac{1}{2}{x^*}^2\} \int_{0}^{\infty} (x^* + t)^q\exp\{-x^*t\}\exp\{-\frac{1}{2}t^2\} g(t) \,dt \nonumber \\
& =  \beta_0{x^*}^{q-1-\beta}\exp\{-\frac{1}{2}{x^*}^{2}\} \int_{0}^{\infty} (1+\frac{u}{{x^*}^2})^q\exp\{-u\}\exp\{-\frac{1}{2}(\frac{u}{x^*})^2\}\left(u^{\beta_1} + o\left(u^{\beta_1}\right)\right) \,du, \end{align} where the last line is obtained after the variable substitution $u = z_1^* t$. Now, since Lebesgue's dominated convergence theorem~\citep{bil95} assures us that $$\lim_{x^* \to \infty} \int_{0}^{\infty} (1+\frac{u}{{x^*}^2})^q\exp\{-u\}\exp\{-\frac{1}{2}(\frac{u}{x^*})^2\}\left(u^{\beta_1} + o\left(u^{\beta_1}\right)\right) \,du = \int_{0}^{\infty} \exp\{-u\}u^{\beta_1} \,du = \Gamma(\beta_1+1),$$ we see from (\ref{I1red}) that $I_1(x^*) \sim \beta_0\Gamma(\beta_1+1){x^*}^{q-1-\beta_1}\exp\{-\frac{1}{2}{x^*}^{2}\}.$ This concludes the proof of part (i) of the theorem. The proof of part (ii) follows similarly.

\end{proof}

As is usually done when analyzing estimators of the type $\estOpt$, we next present a result that characterizes the rate at which the quantity of interest $\alpha$ tends to zero in the asymptotic regime $z_1^* \to \infty$. This will be followed by a result that characterizes the behavior of $\estOpt$.

\begin{theorem}\label{thm:alpharate} Let Assumptions\ref{ass:reform} --\ref{curvcoeffknown} hold. Recalling that $h(\bz) \defn \sum_{i=1}^{p} \gamma_i z_1^i\prod_{j=2}^d z_j^{c_{j}(i)}$, as $z_i^* \to \infty$, $$ \alpha \,\, = \,\, \int_{z_1^*}^{\infty} \phi(z_1) \int_{\Sz(z_1-z_1^*)} h(\bz) \prod_{i=2}^d \phi(z_i) \,\, \sim \,\, (\frac{1}{2\pi})^{d/2}\Gamma(s+1) \eta {z_1^*}^{p-1-s}\exp\{-\frac{1}{2}{z_1^*}^2\}.$$ \end{theorem}

\begin{proof}{Proof.} Write \begin{align}\label{alpharate} \alpha &=  \int_{z_1^*}^{\infty} \phi(z_1) \int_{\Sz(z_1-z_1^*)} h(\bz) \prod_{i=2}^d \phi(z_i),\nonumber \\
&= (\frac{1}{2\pi})^{d/2} \sum_{i=1}^p \gamma_i \int_{z_1^*}^{\infty} z_1^i \exp\{-\frac{1}{2}z_1^2\}\int_{\Sz(z_1-z_1^*)} \prod_{j=2}^d z_j^{c_j(i)} \exp\{-\frac{1}{2}\sum_{j=2}^d z_j^2\} \nonumber \\
& \sim (\frac{1}{2\pi})^{d/2}  \int_{z_1^*}^{\infty} z_1^p \exp\{-\frac{1}{2}z_1^2\}\int_{\Sz(z_1-z_1^*)} \gamma_p\prod_{j=2}^d z_j^{c_j(p)} \exp\{-\frac{1}{2}\sum_{j=2}^d z_j^2\} \nonumber \\
& = (\frac{1}{2\pi})^{d/2} \int_{z_1^*}^{\infty} z_1^p \exp\{-\frac{1}{2}z_1^2\}\tilde{v}_p(z_1-z_1^*), \end{align} where $\tilde{v}_p(z_1-z_1^*) \defn \int_{\Sz(z_1-z_1^*)} \gamma_p \prod_{j=2}^d z_j^{c_j(p)} \exp\{-\frac{1}{2}\sum_{j=2}^d z_j^2\}$. The right-hand side of (\ref{alpharate}) is thus in a form that allows invoking part (i) of Lemma \ref{lem:assrates}, if we can identify the expansion of $\tilde{v}(t)$ around $t=0$.

Towards identifying the expansion of $\tilde{v}_p(t)$ around $t=0$, we notice that \begin{align} \tilde{v}_p(t) &\defn \int_{\Sz(z_1-z_1^*)} \gamma_p \prod_{j=2}^d z_j^{c_j(p)} \exp\{-\frac{1}{2}\sum_{j=2}^d z_j^2\} \nonumber \\ &= \int \gamma_p \mathbb{I}\{{\bf{z}} \in \Sz: z_1 = z_1^* + t\} \prod_{j=2}^d z_j^{c_j(p)} \exp\{-\frac{1}{2}\sum_{j=2}^d z_j^2\}.\end{align} Next, we notice that $\tilde{v}_p(t)$ and $v_p(t) = \int \gamma_p\ind{\mathbf{z} \in \Sz: z_1 = z_1^* + t} \left(\prod_{j=2}^d z_j^{c_{j}(p)}\right)$ defined in Assumption\ref{curvcoeffknown} have the same asymptotic expansion $\eta t^s + o(t^s)$ around $t=0$. We thus invoke part (i) of Lemma\ref{lem:assrates} to conclude that $\alpha \sim (2\pi)^{-d/2} \Gamma(s+1) \eta {z_1^*}^{p-1-s}\exp\{-\frac{1}{2}{z_1^*}^2\}. \qed$
 
\end{proof}

We are now ready to present the main result that characterizes the behavior of the proposed estimator $\estOpt$. Since $\estOpt$ is a biased estimator, the result that follows characterizes both the second moment and the bias of $\estOpt$ in the translation asymptotic regime $z_1^* \to \infty$. 

\begin{theorem}\label{thm:mserate} Recall that $$ \estOpt = (\frac{1}{2\pi})^{(d-1)/2}\frac{\phi(W)}{f_{\lambda(z_1^*)}(W)} \eta {z_1^*}^p (W-z_1^*)^s,
$$ where $f_{\lambda(z_1^*)}(x) = \lambda(z_1^*)e^{-\lambda(z_1^*)x} \mathbb{I}(x \geq z_1^*)$ is the shifted-exponential density function, $W$ is a random variable having the density $f_{\lambda(z_1^*)}(x)$, and the constants $\eta,s$ are from Assumption\ref{curvcoeffknown}. Let Assumptions\ref{ass:reform} -- \ref{curvcoeffknown} hold. Then, if we choose $\lambda(z_1^*) = \theta z_1^*, \theta \in (0,2)$, the following assertions are true as $z_1^* \to \infty$. \begin{enumerate} \item[L.1.] $\E[\estOpt^2] \,\, \sim \,\, (2\pi)^{-d}{\theta}^{-1}(2-\theta)^{-2s-1}\Gamma(2s+1)\eta^2{z_1^*}^{2p-2s-2}\exp\{-{z_1^*}^{2}\}.$ \item[L.2.] $\E[\estOpt - \alpha] \,\,=\,\, o(\alpha).$ \item[L.3.] $\alpha^{-2}\E[(\estOpt - \alpha)^2] \,\, \sim \,\, (\Gamma^{2}(s+1))^{-1}{\theta}^{-1}(2-\theta)^{-2s-1}\Gamma(2s+1)-1.$
\end{enumerate} \end{theorem}

The assertion in $L.3$ implies that $\estOpt$ achieves bounded relative error for any choice of $\theta \in (0,2)$, although it is evident that the specific choice of $\theta$ can have an important effect on efficiency. It is clear from the expression for the relative mean squared error in $L.3$ that the best possible convergence rate for $\estOpt$ will be obtained by maximizing  the function
$\lambda(\theta) \defn \theta(2-\theta)^{2s+1}$ with respect to $\theta$. The following theorem states this formally. A proof is not given here but follows from the basic calculus. \begin{theorem}\label{thetachoice} The asymptotic rate $b_{\estOpt}(\theta) \defn \lim_{z_1^* \to \infty} \alpha^{-2}\E[(\estOpt - \alpha)^2]$ is minimized (with respect to $\theta$) at $\theta^* \defn (1+s)^{-1}$. The corresponding minimal rate is $$b_{\estOpt}(\theta^* ) = (s+1)\left(\frac{s+1}{2s+1}\right)^{2s+1} \frac{\Gamma(2s+1)}{\Gamma(s+1)\Gamma(s+1)}.$$ Furthermore, the minimal rate satisfies $$b_{\estOpt}(\theta^*) \sim \sqrt{s+1} \frac{e}{2\sqrt{\pi}} \mbox{ as } s \to \infty.$$\end{theorem}

\begin{figure}[!h]
	\centering
	\subfigure{
		\includegraphics[trim=0cm 0cm 0cm 2cm, clip=true,width=0.47\textwidth]{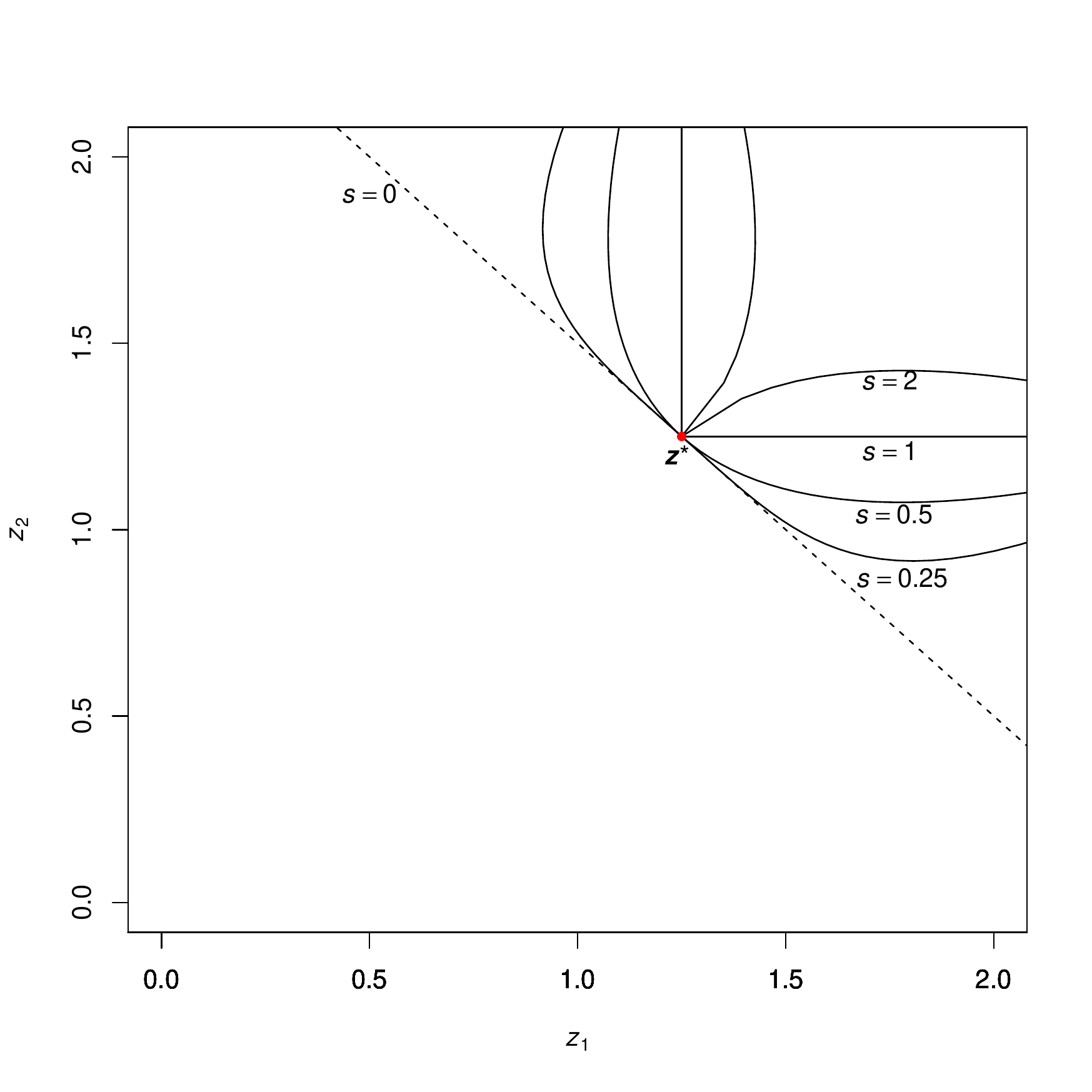}
		\label{fig:sets}
	}
	\subfigure{
		\includegraphics[trim=0cm 0cm 0cm 2cm, clip=true, width=0.47\textwidth]{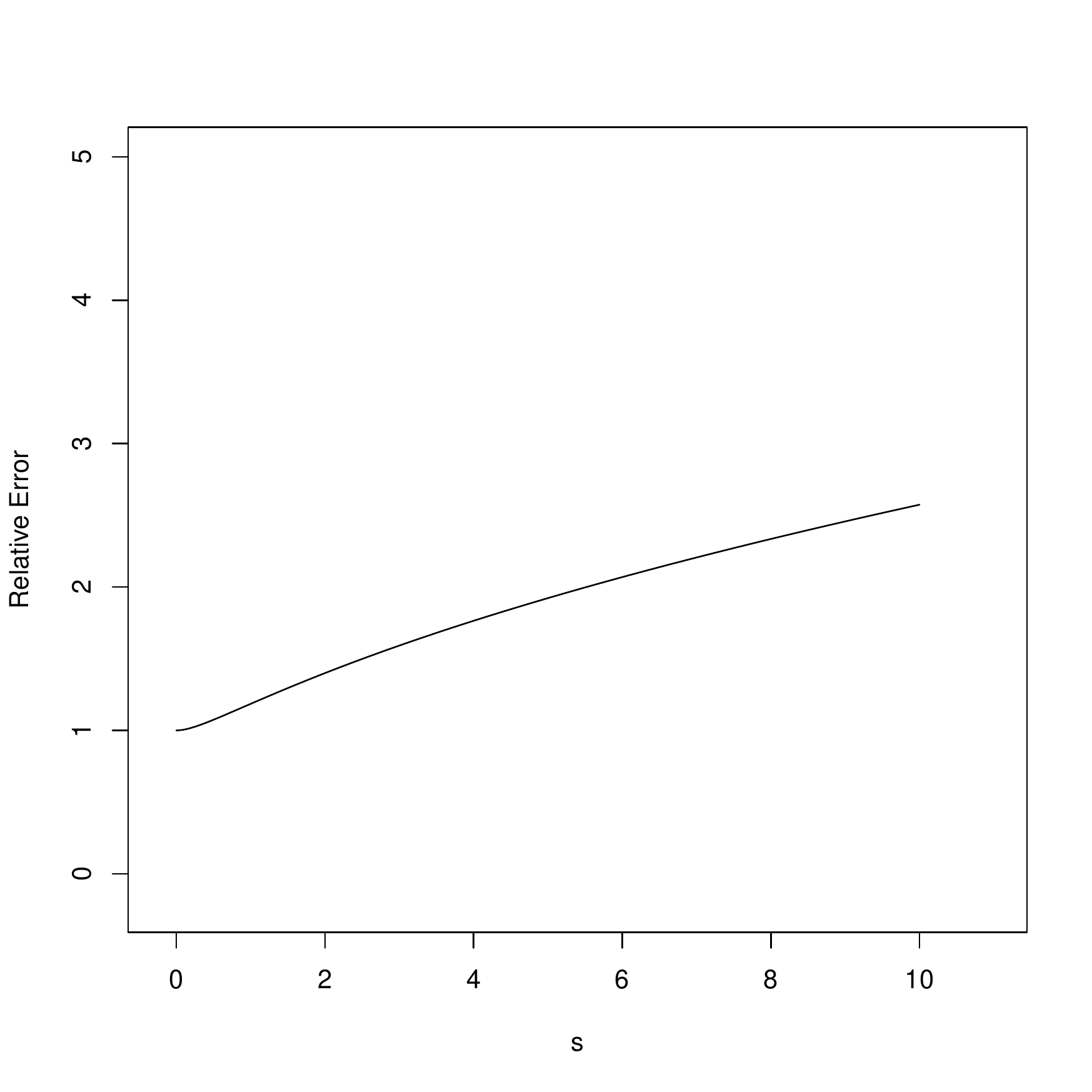}
		\label{fig:relative_errors}
	}
	\caption{The relationship between the structure of the feasible set and the optimal rates achieved through the estimator $\estOpt$. The left panel shows various sets $\Sz$ along with corresponding values of the structural constant $s$. The right panel plots relative error achieved by $\estOpt$ as a function of $s$.}\label{fig:bestrates}
\end{figure}

Theorem \ref{thetachoice} has interesting consequences since it connects the local structure of the feasible set with the convergence rate, with ``sharp" sets calling for lower intensities and ``shallow" sets calling for higher intensities. We have attempted to depict this through Figure~\ref{fig:bestrates}, where various sets $\Sz$ with different structural constants are plotted on the left panel. As the second part of Theorem \ref{thm:mserate} notes, the asymptotic optimal rate $b_{\estOpt}(\theta^*)$ diverges weakly with $s$, as shown by the right panel in Figure~\ref{fig:bestrates}.

\subsubsection{Identifying the Structural Constants $\eta$ and $s$.}\label{sec:etaands}
The reader will recognize that the estimator $\estOpt$ crucially depends on the dominating point $\dpz$, the structural constants $\eta$, $s$, and the highest power $p$ appearing in the polynomial cost function $h$. Of these, the constant $p$ is known since we have assumed that the cost function is given as part of the problem; and the dominating point $\dpz$ can usually be identified in the Gaussian context. The identification of the structural constants $\eta$ and $s$, on the other hand, might involve some effort. In what follows, we demonstrate how the structural constants $\eta$ and $s$ can be identified in a systematic way for a reasonably large class of problems. We first present an example that will lead to the more general approach.

\begin{example}\label{structex} Suppose $\alpha(z_1^*) = \mathbb{E}[(Z_1^aZ_2^b
Z_3^c) \mathbb{I}\{(Z_1,Z_2,Z_3) \in \Sz],$ where $(Z_1,Z_2,Z_3)$ is
the standard trivariate normal and the set $\Sz$ is given by $\Sz
:= \{(z_1,z_2,z_3) \in \real^3: z_1 \geq z_1^* + z_2^6 + z_3^8 +
z_2^4z_3^2\}.$ Suppose $\Sz(z_1-z_1^*) := \{(z_2,z_3): z_2^6 + z_3^8 +
z_2^4z_3^2\leq z_1 - z_1^*\}$, we see then that
\begin{align*}
\alpha(z_1^*) & = \frac{1}{(2\pi)^{3/2}}\int_{z_1^*}^{\infty}
z_1^a \exp\{-\frac{1}{2}z_1^2\}  \int_{\Sz(z_1 - z_1^*)} \,
z_2^bz_3^c\exp\{-\frac{1}{2}(z_2^2 + z_3^2)\}. \nonumber
\end{align*}
Let $t=(z_1-z^*_1)$ and substitute $\tilde{z}_2 = z_2 / t^{x_2}$ and $\tilde{z}_3 =
z_3 / t^{x_3}$ to get the inner integral in the form
\begin{align}\label{gexscale}
g(t) & \defn
t^{(b+1)x_2\,\,+\,\,(c+1)x_3} \quad\quad\int_{\mathcal{S}_{\tilde{z}_1^*,\tilde{\bz}}(1)} \tilde{z}_2^b\tilde{z}_3^c \exp\{-\frac{1}{2}(t^{2x_2}\tilde{z}_2^2 + t^{2x_3}\tilde{z}_3^2)\} \, d\tilde{z}_2d\tilde{z}_3; \nonumber \\
\mathcal{S}_{\tilde{z}_1^*,\tilde{\bz}}(1) & \defn \{\tilde{\bz} \in \real^{d-1}: t^{6x_2-1}\tilde{z}_2^6 + t^{8x_3-1}\tilde{z}_3^8 +
t^{4x_2+2x_3-1} \tilde{z}_2^4 \tilde{z}_3^2\}.
\end{align}
We see that $x_2$ and $x_3$ in the variable substitution above need to be chosen appropriately so that the integral in (\ref{gexscale}) remains bounded away from $0$ and $\infty$ as $z_1^* \to \infty$. Inspecting the powers in the integrand in (\ref{gexscale}) leads to the following linear program (LP) as a way of determinining $x_2$, $x_3$.
 \[
\begin{array}{rl}
\text{minimize} & (b + 1) x_2 + (c+1) x_3  \nonumber \\
\text{subject to:} & 6x_2 \geq 1, \,\, 8 x_3 \geq 1, \,\, 4 x_2 + 2x_3 \geq 1.  
\end{array}
\]
The optimal solution to the above LP is $x^*_2 = 3/16, x^*_3 = 1/8$ when
$(b+1) / (c+1) \ge 1/2$ and $x^*_2=x^*_3=1/6$ otherwise.  
We thus see that as $t \to
0$, $\eta
= \int_{\mathcal{S}_{\tilde{z}_1^*,\tilde{\bz}}(1)} \tilde{z}_2^b\tilde{z}_3^c$, $g(t) = \eta \,\,t^{\frac{3b+2c+5}{16}} +
o(1)$ for $b \ge (c-1)/2$ and $g(t) = \eta \,\,t^{\frac{b+c+2}{6}} + o(1)$ for $b < (c-1)/2$ . The
exact rate at which $\alpha(z_1^*) \to 0$ as $z_1^* \to \infty$ can
then be calculated by 
invoking Lemma\ref{lem:assrates}. \end{example}

Now let us generalize Example\ref{structex}. Recall from Assumption\ref{curvcoeffknown} that the (d-1)-dimensional cost-scaled volume $$v_p(t) = \int_{\Szcom{z_1^*}(t)} \gamma_p\left(\prod_{j=2}^d z_j^{c_{j}}\right)  \,dz_2 dz_3 \cdots dz_{d},$$ where we have suppressed notation to write $c_j$ in place of $c_j(p)$. Suppose the set $\Szcom{z_1^*}(t)$ takes the form $\Szcom{z_1^*}(t) \defn \{(z_2,z_3,\ldots,z_d) \in \real^{d-1}: t \geq f(z_2, z_3, \ldots,z_d)\},$ where the function $f$ is a known $k$-th degree polynomial given by $f(z_2, z_3, \ldots,z_d) = \sum_{\bnu \in \mathcal{C}} \gamma_{\bnu} \Pi_{j=2}^d z_j^{\nu_j}, \bnu = (\nu_2,\nu_3,\ldots,\nu_d), \mathcal{C} \defn \{(k_1,k_2,\ldots, k_d) \in \mathbb{Z}^d: \sum_{i=1}^d k_i \leq k\}.$ Since we seek an expansion of $v_p(t) = \eta t^s + o(t^s)$ as $t \to 0$, we consider the change of variable $\tilde{z}_j = z_j/t^{x_j}$ and write \begin{equation}\label{vpscaled} v_p(t) = \eta t^{\sum_{j=2}^d (c_i + 1) x_j}, \quad \eta = \int_{\mathcal{S}_{\tilde{z}_1^*,\tilde{\bz}}(t)} \gamma_p\left(\prod_{j=2}^d \tilde{z}_j^{c_{j}}\right)  \,d\tilde{z}_2 d\tilde{z}_3 \cdots d\tilde{z}_{d},\end{equation} where $\mathcal{S}_{\tilde{z}_1^*,\tilde{\bz}}(1) \defn \{(\tilde{z}_2,\tilde{z}_2,\ldots,\tilde{z}_d) \in \real^{d-1}: \sum_{\bnu \in \mathcal{C}} \gamma_{\bnu} t^{\sum_{j=2}^d \nu_jx_j}\Pi_{j=2}^d \tilde{z}_j^{\nu_j} \leq 1\}.$ It is evident from (\ref{vpscaled}) that if the scaling variables $x_j, j=2,3, \ldots,d$ are chosen so that $0 < \eta < \infty$ as $t \to 0$, then we would have identified the needed expansion of $v_p(t)$ about $t=0$. For identifying the scaling variables $x_j, j=2,3,\ldots,d$, as in Example\ref{structex}, we consider the following LP in the decision variables $x_j, j=2,\ldots,d$. \[
\begin{matrix} \tag{LP}
\text{minimize} &: \sum_{j=2}^d (c_j + 1)x_j  \nonumber \\
\text{subject to} &: \sum_{j=2}^d \nu_jx_j \geq 1 \quad \forall \nu \in \mathcal{C}.  \\
\end{matrix}
\] Notice that the coefficient vector of the objective function in Problem LP is the same as those appearing within the exponent of the variable $t$ in (\ref{vpscaled}); similarly, the constraints in Problem LP come from the exponent of the variable $t$ in the definition of the set $\mathcal{S}_{\tilde{z}_1^*,\tilde{\bz}}(1)$. Furthermore, the explicit connection between the structure of the constraints in Problem LP and the definition of the set $\tilde{\mathcal{S}}(t)$ dictates that the set $\mathcal{S}_{\tilde{z}_1^*,\tilde{\bz}}(1)$ converges to a set $\tilde{\mathcal{S}}$ (as $t \to 0$) whose volume remains bounded away from zero. This implies that the structural constants $\eta, s$ are given by $\eta = \int_{\mathcal{\tilde{S}}} \gamma_p\left(\prod_{j=2}^d \tilde{z}_j^{c_{j}}\right)$, $s = \sum_{j=2}^d (c_i + 1) x_j^*$, where $x_j^*, j =2,3,\ldots,d$ is the solution to Problem LP.  

The above method for identifying the structural constants applies when the function $f$ is a polynomial. We conjecture that the presented method can be applied even more generally, when $f$ is smooth in a neighborhood around the dominating point $\bz^*$, admitting the application of Taylor's theorem~\citep{royden1} around the point $\dpz$.

\subsection{The Partial-Information Estimators $\estExp$ and $\estLap$}\label{sec:mvnpartial}
Recall that the estimator $\estOpt$ treated in the previous section is a ``full-information" estimator that achieves bounded relative error through shifting and scaling operations that depend explicitly on (i) the knowledge of the dominating point $\dpz$;  (ii) the knowledge of the local curvature of $\Sz$ at $\dpz$ encoded through the structural constants $\eta,s$; and (iii) the knowledge of a supporting hyperplane to $\Sz$ at $\dpz$. In this section we present two partial-information estimators, $\estExp$ and $\estLap$, for use in the absence of (ii) or (iii). The partial-information estimator $\estExp$ is applicable when (ii) is absent but (iii) is present; the partial-information estimator $\estLap$ is applicable when both (ii) and (iii) are absent.  
 
Suppose we have knowledge of a unique dominating point $\dpz$ and of a supporting hyperplane to $\Sz$ at $\dpz$. The partial-information estimator $\estExp$ is then given by \begin{align}\label{estexp} \estExp = \frac{\phi(\bW)}{f_{\lambda(z_1^*)}(\bW)} h(\bW) \mathbb{I}\{\bW \in \Sz\},\end{align} where $\bW = (W_1, W_2, \ldots, W_d)$, $W_1$ is a random variable having the shifted-exponential density function $f_{\lambda(z_1^*)}(x) = \lambda(z_1^*)e^{-\lambda(z_1^*)x} \mathbb{I}(x \geq z_1^*)$ with $\lambda(z_1^*) = z_1^*$, and $W_2, W_3, \ldots, W_d$ are iid standard Gaussian random variables.

For contexts where we have knowledge of a unique dominating point $\dpz$, but not of a supporting hyperplane to $\Sz$ at $\dpz$, the estimator $\estLap$ is given by \begin{align}\label{estlap} \estLap = \frac{\phi(\bW)}{\tilde{f}_{\lambda(z_1^*)}(\bW)} h(\bW) \mathbb{I}\{\bW \in \Sz\},\end{align} where $\bW = (W_1, W_2, \ldots, W_d)$, $W_1$ is a random variable having the Laplace (or double-exponential) density function $\tilde{f}_{\lambda(z_1^*)}(x) = \frac{1}{2}\lambda(z_1^*)e^{-\lambda(z_1^*)|x - z_1^*|}$ with $\lambda(z_1^*) = z_1^*$, and $W_2, W_3, \ldots, W_d$ are iid standard Gaussian random variables.  

Notice that the partial-information estimators $\estExp$ and $\estOpt$ are identical except due to the manner in which the first dimension $W_1$ is generated. Since $\estExp$ assumes knowledge of a supporting hyperplane at the dominating point, the first dimension $W_1$ in $\estExp$ is generated using a shifted exponential that is supported on $[z_1^*,\infty)$. Since the estimator $\estLap$ makes no assumption about the existence of a supporting hyperplane, it uses a Laplace density having support $(-\infty,\infty)$ to model the first dimension $W_1$ of the importance sampling measure. Such careful choice of domain ensures that both $\estExp$ and $\estLap$ are unbiased estimators of $\alpha$.  

The partial-information estimators $\estExp$ and $\estLap$ differ from the full-information estimator $\estOpt$ in two respects. First, in both $\estExp$ and $\estLap$, notice that the rate parameter $\lambda(z_1^*)$ is chosen to be $\lambda(z_1^*) = z_1^*$; this is as opposed to the $\lambda(z_1^*) = (1+s)^{-1}z_1^*$ choice in $\estOpt$, where the curvature constant $s$ was assumed to be available. Second, in the full-information context, it was possible to approximate the outer integral in (\ref{alphasimple}) by the expression $(2\pi)^{-(d-1)/2}\eta {z_1^*}^p(W_{1}-z_1^*)^s$ again because the structural constants $\eta$ and $s$ were assumed to be known. In the current partial-information context, since the constants $\eta$ and $s$ are unknown, the inner integral (\ref{alphasimple}) is estimated using an aceptance-rejection procedure that needs the generation of the Gaussian random variables $W_2, W_3, \ldots, W_d$. 

These differences lead to an inferior relative error of $\estExp$ and $\estLap$ as compared to the optimal estimator $\estOpt$. In Theorem \ref{expRE} that follows, we characterize the relative errors of the estimators $\estExp$ and $\estLap$ in both the translation and the scaling asymptotic regimes (see Section \ref{sec:asympregimes} for formal definitions). The latter regime will become important
  in~\refsec{sec:norta} where the NORTA vectors are considered. A proof of Theorem~\ref{expRE} follows from the application of Lemma\ref{lem:assrates}, and is provided in
the~\refsec{techproofs}. 


\begin{theorem}\label{expRE}
Let the estimators $\estExp, \estLap$ given in (\ref{estexp}) and (\ref{estlap}) estimate $\alpha=\E[h(\bZ) \,\,\ind{ \bZ\in\Sz}]$, where $h(\bz) = \sum_{i=1}^{p} \gamma_i z_1^i\prod_{j=2}^d z_j^{c_{j}(i)}$ is the polynomial defined in Assumption\ref{curvcoeffknown}. Also, let the structural constants $\eta(\bc)$ and $s(\bc)$ be as defined in Assumpton\ref{curvcoeffknown}, and let the constant $$\kappa \defn \frac{1}{\theta}\frac{(2\pi)^{(d-1)/2}}{(2-\theta)^{s(2\bc)+1}}\frac{\eta(2\bc)}{\eta^2(\bc)}\frac{\Gamma(s(2\bc)+1)}{\Gamma^2(s(\bc)+1)}.$$ Then, the following hold.
\begin{itemize}
\item[(i)] For the translated set $\Sz \defn \{ (z_1, \bz) \,|\,
  (z_1-z_1^*, \bz) \in \Szcom{0} \subset \real^d \}$, as $z^*_1\tndi$, 
$$\frac {\E[\estExp^2]} {\alpha^2} \sim 
  \kappa {z^*_1}^{2s(\bc) - s(2\bc)}; \quad \frac {\E[\estLap^2]} {\alpha^2} \sim 
  4\kappa{z^*_1}^{2s(\bc) - s(2\bc)}.$$ 
\item[(ii)] For the scaled set $\Sz \defn \{ (z_1,\bz)\,|\,
  (z_1/z^*_1, \bz/z^*_1)\in\Szcom{0} \subset \real^d\} $, as  $z^*_1\tndi$, 
$$\frac {\E[\estExp^2]} {\alpha^2} \sim
\kappa{z^*_1}^{2s(\bc) - s(2\bc) + (d-1) } \quad \frac {\E[\estLap^2]} {\alpha^2} \sim
4\kappa{z^*_1}^{2s(\bc) - s(2\bc) + (d-1) }.$$ 
\end{itemize}
\end{theorem}

Theorem \ref{expRE} demonstrates that the etimators $\estExp$ and $\estLap$ exhibit polynomial efficiency, with the latter's relative error being twice that of the former.   While this means that neither of the two estimators $\estExp, \estLap$ achieve bounded relative error, it is evident from the definition in Section \ref{sec:notation} that both $\estExp$ and $\estLap$ achieve logarithmic efficiency. It so happens that the exponential-twisting estimator $\estNml$ also achieves polynomial efficiency but is inferior to that of $\estExp$ and $\estLap$. Specifically, for both the translation and scaling regime, $(\mbox{RE}^2(\estNml)/\mbox{RE}^2(\estLap) \sim z_1^* \sqrt{2\pi} 2^{-s(2\bc) - 1}$ as $z^*_1 \to \infty$. Sharp asymptotics of $\estNml$ are given in the electronic companion, in Section~\ref{sec:exptwist}.

We conclude this section with two examples that are meant to illustrate the performance of the estimators $\estExp$ and $\estLap$. The first example illustrates the performance of $\estExp$ and $\estLap$ on a feasible set that is a polyhedral cone opening outward from the vertex $(z_1^*,0,0,\ldots,0)$ and the second example illustrates the performance of $\estExp$ and $\estLap$ on a feasible set defined by a general smooth boundary.

\begin{example}
Suppose $\alpha=\E[\ind{ \bZ\in\Sz}]$ and the set $\Sz$ is a polyhedral cone~\citep{boy04} ``opening outward" along the $z_1$-axis and having vertex at $\dpz \defn (z_1^*,0,0,\ldots,0)$. Formally, the set $\Sz$ can be defined as $\Sz \defn \{ \bz \in \real^d: \bz - \dpz = \sum_{i=1}^d \theta_i \bz_i, \theta_i \geq 0\}$, where $\bz_1, \bz_2,\ldots, \bz_d$ are some fixed points in $\real^d$ and each of whose $z_1$-coordinate exceeds $z_1^*$. Recalling the cross-sectional set $\Szcom{z_1^*}(t) \defn \{(z_2,z_3,\ldots,z_d): \bz \in \Sz, z_1 - z_1^* = t\}$ and using the variable substitution $z_i = u_i (z_1
- z_1^*)$ in the expression for the cost-scaled volume $v_p(t) = \int_{\Szcom{z_1^*}(t)} \gamma_p \prod_{j=2}^d z_j^{c_{j}}  \,dz_2 dz_3 \cdots dz_d,$ we get
\begin{align}\label{alpha_polycone} 
v_p(t) & \sim \left(\int_{\mathbf{u} \in \mathcal{S}(1)} \prod_{i=2}^d |u_i|^{c_i}
d\mathbf{u}\right) t^{(d-1)+\sum_{i=2}^d c_i}.
\end{align} 
We now read-off the structural constants from (\ref{alpha_polycone}) as $s(\bc) = (d-1) + \sum_{i=2}^d c_i$ and $\eta(\bc) = \int_{\mathbf{u} \in \mathcal{S}(1)} \prod_{i=2}^d |u_i|^{c_i}
d\mathbf{u}$. We can then invoke Theorem \ref{expRE} to see that $\estExp$ and $\estLap$ have relative errors that diverge polynomially as $(z_1^*)^{(d-1)/2}$ for the translation regime and $(z_1^*)^{(d-1)}$ for the scaling regime.  \ProofEnd
\end{example}

\begin{example} Let's now consider the more general context of estimating $\alpha=\E[h(\bZ)\ind{ \bZ\in\Sz}]$ on a feasible set $\Sz$ with cross-sectional set $\Szcom{z_1^*}(t) \defn \{(z_2,z_3,\ldots,z_d): t \geq f(z_2, z_3, \ldots,z_d)\}$ such that the function $f$ is a known $k$-th degree polynomial. Recall from Section \ref{sec:etaands} that the polynomial rate $s(\bc)$
in the volume expansion $v_p(t)$ for sets that have a smooth
polynomial boundary at $z^*_1$ is obtained as the optimal value of the linear program Problem LP. Thus, $s(\bc) =
\sum_{i=1}^2 (c_i+1) q^*_i$, where the $q^*_i, i=1,2,\ldots,d$ is the optimal
solution to Problem LP. If $q^{**}_i, i=1,2,\ldots,d$ represents the optimal solution to
the Problem LP with cost $2\bc$, then we have that
\begin{equation}\label{smoothrate}
2s(\bc) -s(2\bc) = \left(\sum_{i=1}^d (2c_i+1)  q^*_i -  \sum_{i=1}^d (2c_i+1)
q^{**}_i\right) + \sum_{i=1}^d q^*_i \,\,> \,\, 0 \,\,.
\end{equation}
The term in parenthesis on the right-hand side of (\ref{smoothrate}) is non-negative; it is zero if the solution to
Problem LP with cost vector $\bc$ is also optimal to the Problem LP with cost vector
$2\bc$. Also, the second term on the right-hand side of (\ref{smoothrate}) is strictly positive given the constraints of Problem LP. We conclude from Theorem \ref{expRE} that $\estExp$ and $\estLap$ have relative errors that diverge polynomially with exponents given by the right-hand side of (\ref{smoothrate}). \ProofEnd
\end{example}


\section{THE CONSTRAINED NORTA-VECTOR CONTEXT}\label{sec:norta}

Recall our problem statement of having to estimate \begin{equation}\label{alphanorta}
\alpha \defn \E[h(\bX) \,\,\ind{ \bX\in\Sx}],
\end{equation} where the set $\Sx$ has a small implied measure under the distribution of $\bX$. In Section \ref{sec:mvn}, we assumed that the random variable $\bX$ belonged to the Gaussian family, leading to the reformulated problem in (\ref{gauss_pb}). In this section, we generalize to the important setting where $\bX$ is a NORTA vector, that is, $\bX$ has a distribution whose dependence structure is specified through the Gaussian copula~\citep{nel2013,mcnfreemb2005,nelsen1,ghosh1rv}. Since such generalization renders the estimators $\estOpt, \estExp, \estLap$ inapplicable (particularly due to possible violation of Assumption \ref{ass:reform}(b)), we present a fourth estimator $\estNt$ that is derived from $\estLap$ and achieves polynomial efficiency under certain conditions. 

For easing exposition of $\estNt$, let us now reformulate the problem in (\ref{alphanorta}) for the context where $\bX$ is a NORTA random vector. If $\bX$ is a NORTA random vector, it is well-known~\citep{nel2013,ghosh1rv} that the components $X_i, i=1,2,\ldots,d$ of $\bX$ can be expressed through the NORTA map $\bT :
\real^d\rightarrow \real^d$ as $\bX = \bT(\bZ) = {\bF}^{-1} \circ {\bPhi} \circ A^T\bZ$  
for $\bT = (T_1, T_2,\ldots,T_d)$, $\bPhi = (\Phi,\Phi,\ldots,\Phi)$,
$\bF = (F_1,F_2,\ldots,F_d)$, where $\Phi$ is the standard normal
distribution function, $F_i$ is the marginal distribution of
$X_i$ with inverse $F_i^{-1}(p) \defn \{x: F_i(x) \geq p\}$, and $\bZ$
is a standard Gaussian random vector, and the matrix $A$ defines a
Gaussian random vector with $0$ mean and positive-definite covariance matrix $\Sigma_Z = AA^T$. The NORTA inverse map $\bT^{-1}(\bX) = (A^{-1})^T{ \bPhi}^{-1}({\bF}(\bX))$ allows to reformulate the problem in (\ref{alphanorta}) into the Gaussian space as follows.
\begin{align}\label{alphanortareform}
\alpha & \defn \E[h \circ \bT (\bZ) \,\,\ind{ \bZ\in\Sz}]; \nonumber \\
\Sz & \defn \{\bz \in \real^d: \bT(\bZ) \in \Sx\}.
\end{align} We limit our attention to continuous marginal distributions $F_i$ such
that their density function exists and is strictly positive
everywhere. (Note that continuous $F_i$ may have ``flat'' regions,
e.g. where the density is zero, and hence $F_i^{-1}$ and $T_i$ may
have jump discontinuities.) 

Notice that the problem in (\ref{alphanortareform}) has been recast in the Gaussian space implying that one of the estimators $\estOpt$, $\estExp$, or $\estLap$ could be applicable as an estimator. The complication, however, is that the transformed feasible set $\Sz$ may not be well-behaved even if the set $\Sx$ is well behaved. For example, properties such as convexity are generally not preserved in the sense that even if $\Sx$ is convex, $\Sz$ may not be convex. More importantly, the dominating point $\dpz$ that was assumed to be known (and unique) throughout Section \ref{sec:mvn} is neither easily identified nor unique, rendering the estimators $\estOpt$, $\estExp$, $\estLap$ inapplicable for the current context. 

\begin{remark} It so happens that the image of the set of boundary points of $\Sx$ remains on the boundary of $\Sz$. One may thus reasonably expect that
the pre-image of the dominating set $\min_{\Sx} \|\bx\|_2$ is ``close" to the
dominating set $\arg\min_{\Sz}\|\bz\|_2$ in the $\bZ-$space. This and some other properties are summarized in Section \ref{sec:tprop} appearing in the appendix. \end{remark}

\subsection{The ecoNORTA Estimator $\estNt$}\label{sec:econortaestimator}

Considering that the set $\Sz$ in the NORTA context (\ref{alphanortareform}) can be poorly behaved, in what follows, we assume only that the dominating set $\dpzs=\{\dpz \,:\,\, \dpz=
\arg\inf_{\bz\in\Sz} \|\bz\|_2\}$ exists. Specifically, we make no assumptions about the cardinality of the set $\dpzs$ except that $|\dpzs| < \infty$, nor about the local curvature of $\Sz$ at the dominating set $\dpzs$ (as embodied in Assumption~\ref{curvcoeffknown}).

The \ent~estimator $\estNt$ is a mixture of $\estLap$ estimators, each of which is centered at the elements of $\dpzs$. Formally, we write
\begin{equation}\label{defn:econorta}
\estNt =\,\,\sum_{k=1}^{|\dpzs|} \left\{ \ind{{\bf K}=k}
\,\,\frac{\phi(\bW)}{f_{z_1^*}(\bW)}\quad h\circ T (R_k \bW)  \quad \ind{ R_k
  \bW \,\,\in\,\, \Sz}\right\}, 
\end{equation}
where the random variable $\bW = (W_1, W_2, \ldots, W_d)$, $W_1$ is a
random variable having the shifted-Laplace density function defined in~\refsec{sec:mvnpartial},
 and $W_2, W_3, \ldots, W_d$ are iid standard Gaussian random
variables. The mixing random variable ${\bf K}$ has a multinomial distribution supported on $\dpzs = \{\bz^*_1, \bz^*_2, \ldots, \bz^*_k\}$ with probability mass function $\P({\bf K} = \bz^*_k)= \nu_k,\,\,k=1,\ldots,|\dpzs|$.

As can be seen in (\ref{defn:econorta}), the estimator $\estNt$ co-opts the estimator $\estLap$ for a context where the set $\Sz$ may have multiple dominating
points. The multinomial random variable ${\bf K}$ in (\ref{defn:econorta}) samples one such dominating point, and the rotation matrix $R_k$ reverses alignment with the $z_1-$axis to generate a random variate $R_k\bW$ near the
$k$-th dominating point $\bz^*_k$. The unbiasedness of $\estNt$ in
estimating $\alpha$ follows from the absolute continuity of the
Laplace density with respect to the standard normal density $\bPhi$. As we will demonstrate shortly through Theorem~\ref{thm:econortaworks}, under certain structural conditions, the estimator $\estNt$ achieves the
same slow polynomially growing relative error as $\estLap$. This compares favourably with the exponential
rates of the na\"ive estimator $\estAR$, and also the exponential-twisting estimator $\estNml$ which is analyzed in Section~\ref{sec:exptwist}.
For showing the polynomial efficiency of the estimator $\estNt$, we make the following regularity assumption on $\Sz$. 

\begin{assumption}\label{ass:geom}
  \begin{enumerate}
  \item[\rm{(a)}] At each $\bz^*_k\in\cZ^*$, the set $\Sz$ locally
    satisfies Assumption\ref{curvcoeffknown} with parameters
    $\eta_k(\bc), s_k(\bc)$ for cost exponents
    $\bc=(c_2,\ldots,c_d)$.
  \item[\rm{(b)}] The half-spaces $\cH(\bz^*_k) = \{\bz \, : \,\bz^t
    \bz^*_k \geq \mathbf{z^*_k}^t \bz^*_k\}$ for all $\bz^*_k\in\cZ^*$
    are such that $\Sz\subset\mathop{\cup}_{k=1}^{|\cZ^*|}
    \cH(\bz^*_k)$.
  \end{enumerate}
\end{assumption}

\begin{figwindow}[1,r,
\includegraphics[width=0.5\textwidth]{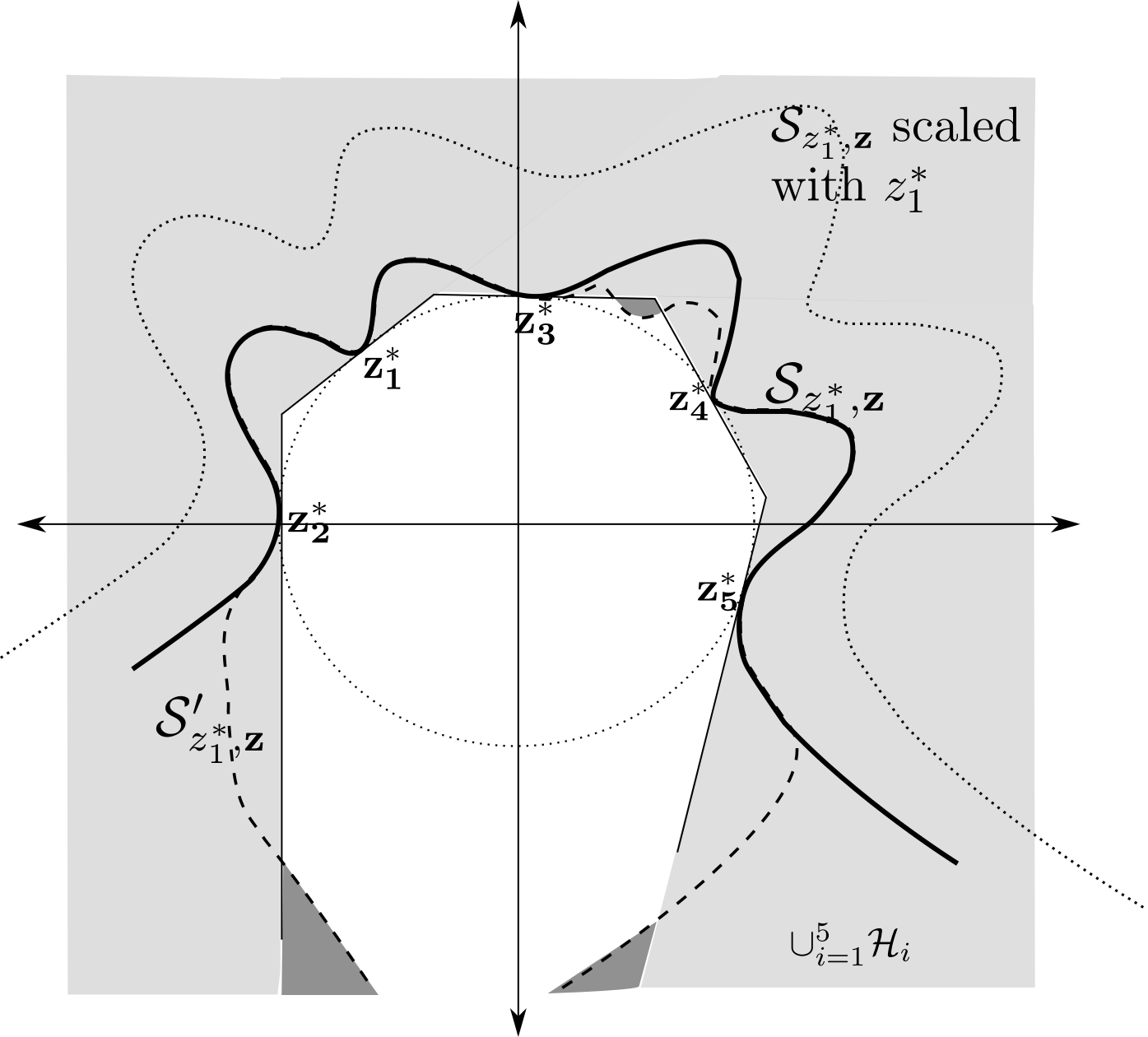},
                {\label{fig:econorta} An illustration of the conditions needed
                  for~\refthm{thm:econortaworks}}]
  In~\reffig{fig:econorta}, the example $\Sz$ set plotted in bold has
  five dominating points in the set $\dpzs$. The set $\dpzs$ locally satisfies
  Assumption\ref{curvcoeffknown} at each dominating point as required in
  Assumption\ref{ass:geom}(a). Assumption\ref{ass:geom}(b) states
  that the set as a whole lies inside the union of the half-spaces
  defined by the tangent hyperplanes at each of the dominating points, as
  shown by the light-shaded region. Also plotted is an example set
  $\Sz'$ that satisfies Assumption\ref{ass:geom}(a) but not
  Assumption\ref{ass:geom}(b) in the dark-shaded regions. So, the
  requirement Assumption\ref{ass:geom}(b) may be violated in practice
  by sets like $\Sz'$, but the regions where these happen are far from
  the dominating points, and so it is reasonable to expect that the
  relative error does not degrade too much. Thus, the uniqueness of
  the dominating point (Assumption\ref{ass:reform}) as well as the assumption about the separating hyperplane (Assumption\ref{supphyp}) are relaxed. Assumption\ref{ass:geom} is inspired by a similar assumption in~\cite{juneja06}.

\end{figwindow}

\begin{theorem}\label{thm:econortaworks}
  Suppose the set $\dpzs=\{\dpz \,:\,\, \dpz= \arg\inf_{\bz\in\Sz}
  \|\bz\|_2\}$ exists, $|\dpzs| < \infty$ and $z^*_1 = \inf\{\|\bz\|:
  \bz\in\Sz\}.$ Let the $(d-1)$-dimensional cost-scaled volume
  $v^k_p(t)$ around $\dpzsub{k}$ (as defined in Assumption\ref{curvcoeffknown})
  satisfy the expansion $v^k_p(t) = \eta_k(\bc) t^{s_k(\bc)} +
  o(t^{s_k(\bc)})$ as $t \to 0$. Further, let
  Assumption\ref{ass:geom} hold and let the $\nu_k, k=1,2,\ldots, |\dpzs|$
  form a probability mass function such that $\nu_k > 0, \forall k$.
  Then, in the scaling regime, as $z^*_1\tndi$, $$\frac {\E[\estNt^2]}
  {\alpha^2(z^*_1)} \sim \max_{k=1,\ldots,K}
  \kappa(\dpzsub{k})\,\,\,\frac 1 {\nu_k}\quad
        {(z^*_1)^{2s_k(\bc)-s_k(2\bc) + (d-1)}},$$ where the
        $\kappa(\dpzsub{k})$, the constants defined
        in~\refthm{expRE}, are calculated at each dominating point
        $\dpzsub{k}$ using the corresponding constants $s_k(\bc)$ and
        $\eta_k(\bc)$.
\end{theorem}

\ProofOf{\refthm{thm:econortaworks}} 


~\refthm{expRE}(ii) will be applied to show the relative error
result. Let $\estLapsub{k}$ represent, for each $k=1,\ldots,|\dpzs|$,
an $\estLap$ estimator centered at the dominating point $\dpzsub{k}$
to estimate $\alpha_k=\E [ h(\bZ) \{\bZ\in \,\,\Sz\cap\cH_k\} ]$. The
\ent~estimator $\estNt$ is a mixture of estimators $\estLapsub{k}$
with mixing probabilities $\nu_k$.  Scaling the original set
$\Szcom{0}$ as $ \Szcom{z^*_1} = \{ (z_1,\bz)\,:\, (z_1/z^*_1 ,
\bz/z^*_1) \in \Szcom{0} \}$ preserves the key geometric requirements
of Assumption~\ref{supphyp} and~\ref{curvcoeffknown}. Denote by $L(z)$
the likelihood ratio
of the $d-$dimensional $\Phi$ and the distribution of $\estNt$, and
likewise $L_k(z)$ as the likelihood ratio of the $\estLapsub{k}$
estimator restricted to the set $\Sz\cap\cH_k$. Split the
second-moment of $\estNt$ as (with $K=|\dpzs|$) 
  \begin{eqnarray*}
\E[\estNt^2] &=&\int_{\Sz} L(z)\,\,d\Phi(z) \quad
= \quad \int_{\Sz\cap\cH_1}L(z)\,\,d\Phi(z) + \int_{\Sz\cap
  (\cH_2\setminus\cH_1)}L(z)\,\,d\Phi(z) + \quad\ldots \\
& & \qquad\qquad+\quad\int_{\Sz\bigcap\left(\cH_{K}\setminus
  \left(\bigcup_1^{K-1}\cH_j\right)\right)}L(z)\,\,d\Phi(z).
  \end{eqnarray*}
Each set in the partition is non-empty given the structure of $\cH_k$
in~Assumption~\ref{ass:geom}(b).
For any set in the partition, $L(z) \le L_k(z) / \nu_k$ for any $k$.
The tightest bound for the $\ell-$th set  $\Sz\cap
\,(\,\cH_{\ell}\,\setminus\,(\cup_1^{\ell-1}\cH_j)\,)$ 
in the partition is obtained by using the $L_{\ell}(z)$ term.
Thus, we get that
  \begin{eqnarray*}
\E[\estNt^2] &\le& \frac 1 {\nu_1}\int_{\Sz\cap\cH_1}L_1(z)\,\,d\Phi(z) +
\frac 1 {\nu_2}\int_{\Sz\cap 
  (\cH_2\setminus\cH_1)}L_2(z)\,\,d\Phi(z) + \ldots \\
& &\qquad\qquad\qquad\qquad\qquad\qquad\qquad\qquad\qquad
+\frac 1 {\nu_K}\int_{\Sz\bigcap\left(\cH_{K}\setminus
  \left(\bigcup_1^{K-1}\cH_j\right)\right)}L_K(z)\,\,d\Phi(z) \\
&\le& \frac 1 {\nu_1}\int_{\Sz\cap\cH_1}L_1(z)\,\,d\Phi(z) + \frac 1 {\nu_2}\int_{\Sz\cap
  \cH_2}L_2(z)\,\,d\Phi(z) + \ldots
+ \frac 1 {\nu_K}\int_{\Sz\cap\cH_{K}}L_k(z)\,\,d\Phi(z) \\
&=& \sum_k \frac {\E[\estLapsub{k}^2]} {\nu_k} \quad\le\quad
\sum_k \kappa(\dpzsub{k})\,\,
\frac {(z^*_1)^{2s_k(\bc) - s_k(2\bc) + (d-1)} \alpha_k^2} {\nu_k} \\
& \le &
\alpha^2  \sum_k \kappa(\dpzsub{k})
\frac {(z^*_1)^{2s_k(\bc) - s_k(2\bc) + (d-1)} } {\nu_k}. 
\end{eqnarray*}
The final inequality follows by noting that
$\alpha_k\le\alpha,\,\,\forall k$. The penultimate inequality follows
from~\refthm{expRE}(ii), which we can apply as per Assumption~\ref{ass:geom}(a).
\ProofEnd


It can be shown that for the polynomial growth of relative error to
be preserved for polynomial cost functions in the NORTA context, the
implied cost function $h \circ \bT (\bz) $ should be regularly varying~\citep{res1987}.  
This can be expected to hold when the marginal distributions
$F_i$ exhibit exponential tails, such as the exponential, gamma and
phase-type distributions. Consider, for example, a NORTA vector having
exponential marginals with rate $\lambda_i$ for the $i$-th
component. Then, $\bT_i(\bz) = -\lambda_i^{-1}\log (\bar{\Phi} ({y}_i))
\sim \lambda_i^{-1}\left(\frac {{y_i}^2} 2 + \log y_i + \frac 1 2 \log
(2\pi)\right)$ as $\bz\tndi$, where ${\bf y} = A^T\bz$ and using the
asymptotic $\bar{\Phi}(y) \sim \phi(y)/ y$ as $y\tndi$
\citep{johnson1}. Thus, $h\circ \bT(\bz)$ is regularly varying in
$\bz$.

This is however not the case for marginals with heavier than
exponential tails. Consider a Pareto distribution $F_i(x) = 1-
x^{-\alpha}$, where  the shape parameter $\alpha>2$. Then,
$\bT_i(\bz) = (\bar{\Phi}(y_i))^{-1/\alpha} \sim (y_i e^{y_i^2/2}
\sqrt{2\pi} )^{1/\alpha}$ as $\bz\tndi$ where again ${\bf y} =
A^T\bz$. Here, $h\circ \bT(\bz)$ is super-exponential in $\bz$.

We end this section by speculating that the $\estLap$ estimator in $\estNt$ could be replaced with the estimator $\estExp$, and or even the estimator $\estOpt$, if it is (somehow) known that the needed structural conditions for $\estExp$ and $\estOpt$ hold locally. For incorporating the $\estOpt$ estimator within $\estNt$, we will of course need to
know the curvature information $s_k(\cdot)$ at each dominating point
$\dpzsub{k} \in \dpzs$.

\subsection{The ecoNORTA Algorithm for Implementing $\estNt$}\label{sec:econortaalgoritm} The estimator $\estNt$ proposed in Section~\ref{sec:econortaestimator} is usually unimplementable simply because the set of dominating points $\dpzs$ is unknown. A natural idea, however, is to estimate the set $\dpzs$ sequentially and use it within the expression (\ref{defn:econorta}), resulting in an ``implementable" version of $\estNt$. 
The ecoNORTA algorithm, listed as Algorithm 1, provides details on the construction of such a sequential version of $\estNt$. The ecoNORTA algorithm, while simple in principle, has a number of heuristic steps introduced for enhanced and stable implementation. In what follows, we describe these steps. 

The ecoNORTA algorithm is an iterative algorithm which, during each iteration $k$, maintains an estimate $\dpzsn{k}$ of the dominating set $\dpzs$, and a set $\cN_k$ that is best described as the set of ``candidate" dominating points. The candidate set $\cN_k$ is chosen as the collection of a fixed number ($\Delta$) of ``nearest points" from amongst those feasible points generated by iteration $k$. The ecoNORTA algorithm then repeats Steps 5 and 6 in Algorithm 1 a total of $m_k$ times during each iteration $k$. In Step 5, mimicking the $\estNt$ estimator given in (\ref{defn:econorta}), the current estimate $\dpzsn{k}$ of the dominating set is used to generate an observation from a multinomial mixture of Laplace distributions each of which is centered on the elements of $\dpzsn{k}$. The random variate generated in Step 5 is then used to update the candidate set $\cN_k$ in Step 6. In Step 7, a clustering algorithm, e.g.,~\cite{lloyd82}, is used to identify at most $C$ clusters of points in the candidate set $\cN_k$; clusters that are too close in the sense that their centroids $\zeta(\cC) \defn \frac 1
{|\cC|} \sum_{z\in\cC} z$ are within distance $\delta>0$ are merged into a
single cluster. The maximum number of clusters $C$ and the minimum
distance $\delta$ between clusters are both algorithm parameters. Points closest to the origin from each of the identified clusters in
Step 7 are deemed dominating point estimates and added to the set
$\dpzsn{k}$ in Step 8. As noted, the two steps 5 and 6 are repeated $m_k$ times during the $k$-th iteration of ecoNORTA towards constructing a sequence of dominating set estimators $\{\dpzsn{k}\}$. An initial estimate $\dpzsn{0}$ of $\dpzs$ can be obtained via \reflem{lem:szprop}, and the algorithm parameter $m_k$ is usually set to a convenient constant during implementation. 

While the procedure in Algorithm 1 is expected to consistently estimate all the optimal solutions to $\arg\min_{\Sz}\|z\|_2$, as a practical matter, the updating of the set $\dpzsn{k}$ stops when successive sets $\dpzsn{k}$, $\dpzsn{{k+1}}$ change very little. The resulting dominating set estimate $\dpzsn{k}$ upon such termination is used in place of $\dpzs$ in (\ref{defn:econorta}) for constructing the ecoNORTA estimator.  

\begin{figure}[t!]
  \large
   \includegraphics[trim=0 374 0 75, clip=true,width= \textwidth]{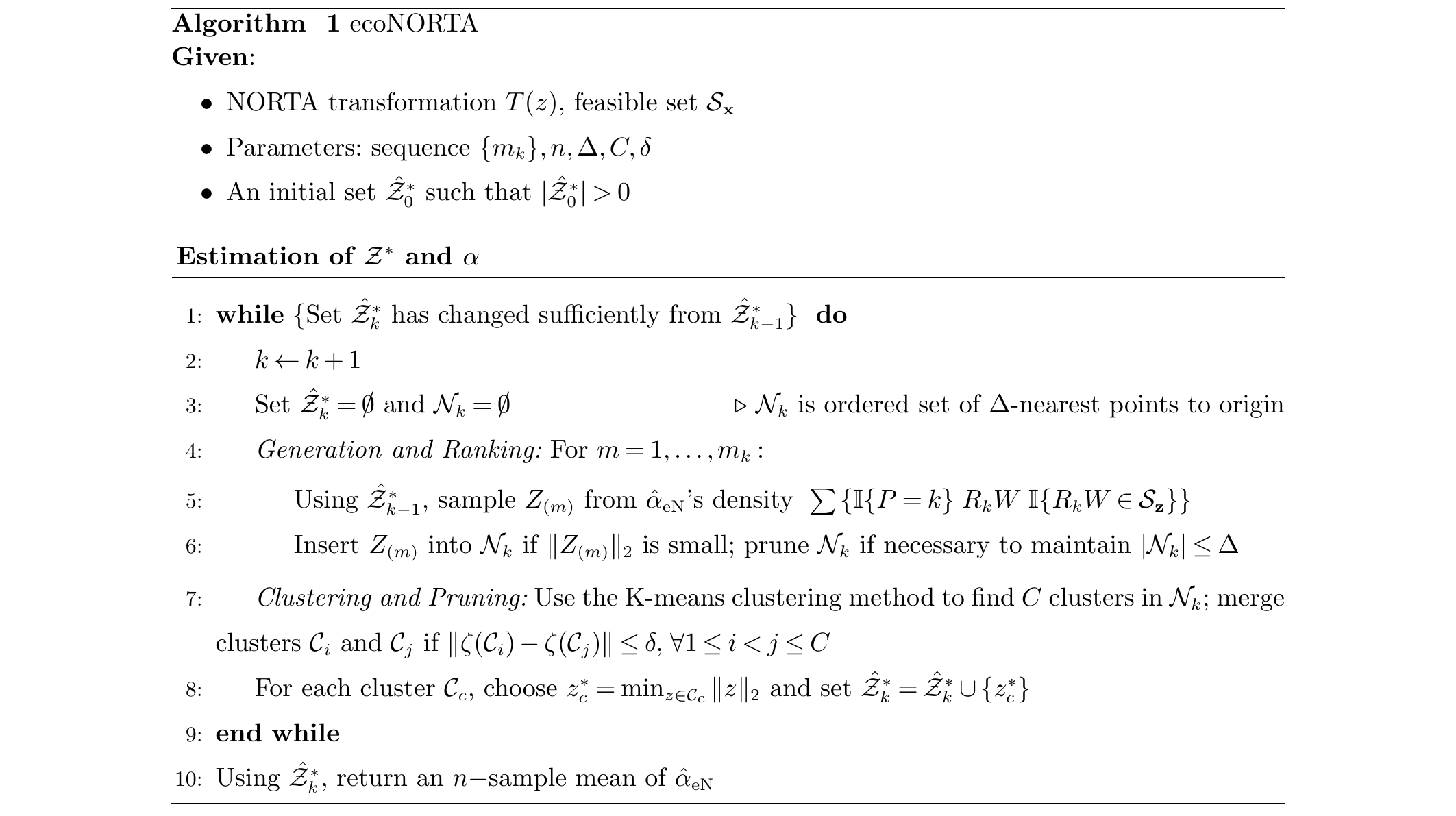}
   \label{fig:Alg1}
\end{figure}

It is important to note that Algorithm 1 is essentially a mechanism to make $\estNt$ implementable, since the dominating set $\dpzs$ appearing in (\ref{defn:econorta}) is unknown. The problem of identifying the set $\dpzs$ of {\em all} points
nearest to the origin, that is, all solutions to $\min_{\Sz}\|z\|_2$, is a
deterministic global optimization problem. An extensive literature~\citep{pardalosromeijn1}
exists on solution approaches to such formulations, and any reasonable procedure, e.g.,  \cite{2016gramecy}, to seek all global optima can be used within Algorithm 1. A fuller analysis of this question takes us beyond the scope of this paper.

\section{Numerical Experiments}\label{numerics}
In this section, we present a sampling of our fairly extensive numerical experience with the estimators presented in this paper. In Section~\ref{sec:mvnexpts} we present an example in the Gaussian context aimed at showcasing the power of the full-information estimator $\estOpt$. The experiment in Section~\ref{sec:mvnexpts} is designed so as to allow the numerical computation of the exact value of $\alpha$ in any number of dimensions. In Section~\ref{sec:mvnexpts}, we compare $\estOpt$ against the partial-information estimator $\estExp$, the exponential-twisting estimator $\estNml$, and the acceptance-rejection estimator $\estAR$. 

Section~\ref{sec:nortaexpts}, which follows Section~\ref{sec:mvnexpts}, illustrates the power of the ecoNORTA estimator throug a configure-to-order manufacturing system example. Since none of the other estimators are applicable in this context, the ecoNORTA estimator $\estNt$ (as implemented via the ecoNORTA algorithm) is compared against the acceptance-rejection estimator $\estAR.$

\subsection{A Constrained Gaussian Experiment}\label{sec:mvnexpts}

  Suppose that $\alpha(z_1^*) =
  \mathbb{P}\{\mathbf{X} \in \Sz\},$ where $\mathbf{X}$
  is the standard Gaussian vector and
$\Sz := \{(z_1^*,\mathbf{z}) \in \real \times
  \real^{d-1}: z_1 \geq z_1^* +
  \frac{1}{2}\mathbf{z}^TA\mathbf{z}\},$ where the positive-definite
  matrix $A = Q^T\Lambda^2 Q$ for a strictly positive diagonal
  $\Lambda$. Denote $B(\mathbf{0},1) \defn \{\mathbf{z}':
  \frac{1}{2}\mathbf{z}'^T \mathbf{z}' \leq 1\}$ and $\eta =
  \int_{\mathbf{z}' \in B(\mathbf{0},1)}d\mathbf{z}'$. Setting
  $\mathbf{z}' = (z_1-z_1^*)^{-1/2}\Lambda Q \mathbf{z}$ and $R =
  (Q\Lambda)^{-1}$, we get
      \begin{align*}
        \alpha(z_1^*) & = \frac{1}{(2\pi)^{d/2}}\int_{z_1^*}^{\infty}
        \exp\{-\frac{1}{2}z_1^2\}\, dz_1 \int_{\mathbf{z}' \in
          B(\mathbf{0},1)}
        (z_1-z_1^*)^{(d-1)/2}|\Lambda^{-1}Q^{-1}|\exp\{-\frac{z_1-z_1^*}{2}\mathbf{z}'^TR^TR\mathbf{z}'
        \}\, d\mathbf{z}'  \\
& \sim \frac{|\Lambda
          Q|^{-1}}{(2\pi)^{d/2}}\,\,\eta\,\,\int_{z_1^*}^{\infty}
        \exp\{-\frac{1}{2}z_1^2\} (z_1-z_1^*)^{(d-1)/2}\, dz_1 \,\,\,\,=\,\,\,\,
        \frac{|\Lambda Q|^{-1}}{(2\pi)^{d/2}}\,\eta\,
        \Gamma(\frac{d+1}{2}){z_1^*}^{-\frac{d+1}{2}}\exp\{-\frac{1}{2}{z_1^*}^2\}.
      \end{align*}


%

\noindent Let $A=2I$, where $I$ is the identity matrix. From~(\ref{estopt}), the
optimal estimator $\estOpt$ for this example is
\[ \estOpt(z_1^*) = \frac{\phi(Z)}{f_{\lambda(z_1^*)}(Z)}
\frac{1}{2^{(d-1)/2}\Gamma \left(\frac{d-1}{2}+1\right)}
(Z-z_1^*)^{(d-1)/2}, \] where $Z$ is a shifted-exponential random
variable with density
$f_{\lambda(z_1^*)}(z_1)=\lambda(z_1^*)e^{-\lambda(z_1^*)(z_1-z_1^*)} \mathbb{I}\{z_1 \geq z_1^*\}$,
and following Theorem \ref{thetachoice}, the optimal intensity
$\lambda(z_1^*)=2z_1^*/(d-1)$. The estimators
$\estExp$,  $\estNml$ and $\estAR$ follow the
definitions~(\ref{estexp}),~(\ref{eq:alphahatnml}) and~(\ref{naivecv}). 
 
Table~\ref{mvnres}(a) displays numerically computed relative mean squared errors
of the estimators $\estOpt$, $\estExp$, $\estNml$, and $\estAR$ for problem dimensions
$d=3$ and $d=9$, and for increasing values of the distance $z_1^*$ of the
dominating point of $\cS(z_1^*)$ from the origin. Recall that since the estimators $\estExp$, $\estNml$, and $\estAR$ are unbiased, the relative mean
squared error $\mbox{MSE}(\hat\alpha)/\alpha^2$ for each of these estimators is 
$\mbox{MSE}(\hat\alpha)/\alpha^2=\mathbb{E}[\hat\alpha^2]/\alpha^2-1$. Since the full-information estimator $\estOpt$ is biased, its relative mean squared error $\mbox{MSE}(\hat\alpha)/\alpha^2=\mathbb{E}[\hat\alpha^2]/\alpha^2 - 1 - 2\E[(\estOpt - \alpha)]/\alpha$. Each cell in the table was computed with a sample size of $10^8$. Although computing $\alpha(z_1^*)$ is generally intractable, the
quadratic form of $\cS(z_1^*)$ in this particular problem lets us
numerically evaluate $\alpha(z^*_1)$ from a recursion, details of
which are provided in ~\refsec{additional_examples}.
 
As predicted by theory, and as can be seen in Table~\ref{mvnres}(a), the performance of the full-information estimator is dramatically better than the partial-information estimator $\estExp$ and the exponential-twisting estimator $\estNml$. The difference becomes particularly pronounced for ``harder" problems, e.g., in higher dimensions and for larger values of $z_1^*$. As an example, in dimension $d=9$ and for $z_1^* = 5$, the mean squared error of the exponential-twisting estimator is about $10,000$ times larger than that of the full-information estimator $\estOpt$. While we have reported results only for dimensions $d=3$ and $d=9$, the performance of the full-infomration estimator $\estOpt$ appears to remain quite stable for much larger $d$ values, even though the corresponding calculation of $\alpha$ values becomes numerically unstable.   

As can be seen from Table~\ref{mvnres}(a), the partial-information estimator $\estExp$ does not perform nearly as well as $\estOpt$ although it seems to generally perform better than the exponential-twisting estimator $\estNml$ and the acceptance-rejection estimator $\estAR$. Such performance is consistent with our asymptotic theory, as demonstrated by Figure~\ref{mvnres}(b) which plots the ratio of the mean squared error of the $\estNml$
and $\estExp$ estimators as $z^*_1$ grows. In each case, the ratio is seen to grow as $O(z^*_1)$, as
predicted by Theorem~\ref{expRE} and
Lemma~\ref{lem:estNmlmulti}. The asymptotics seem to take longer to ``kick-in" as $d$ becomes larger.
 
\begin{table}[t!]
  \centering
  \subfigure[Relative error of $\estOpt$, $\estExp$, $\estNml$, and $\estAR$]{
    \includegraphics[trim=175 480 175 95, clip=true,width= 0.47\textwidth,angle=0]{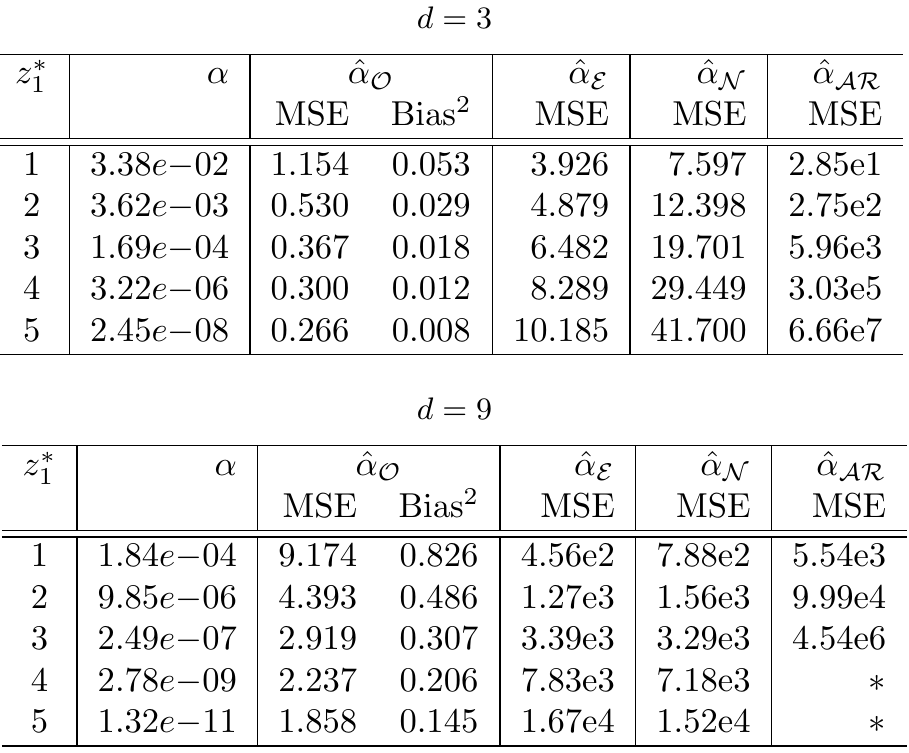}
   \label{fig:mvn_table}
  }
  \subfigure[Log-log plot of relative errors of $\estNml$ and $\estExp$.
  ]{
    \includegraphics[trim=0 130 20 130,clip=true,width=0.48\textwidth]{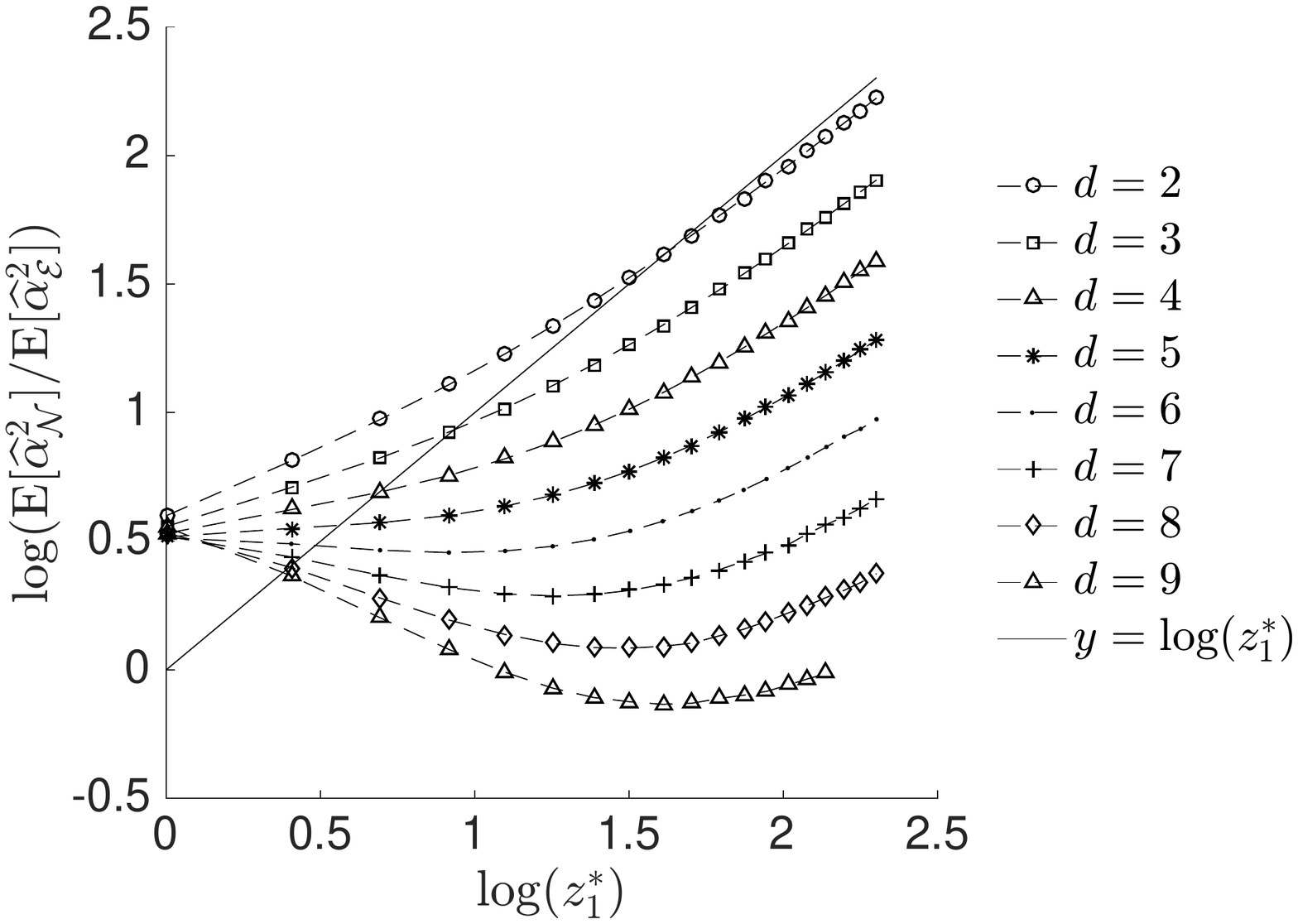}
    \label{fig:mvn_plot}
  }
   \caption{The table on the left lists the relative mean squared error of each
     estimator for dimensions $d=3,9$ and increasing $z_1^*$
     values. The
     plot on the right shows that the ratio of mean squared errors of the exponential-twisting estimator $\estNml$ and the partial-information estimator $\estExp$ grows as $O(z^*_1)$, consistent with theory. \label{mvnres}}

\end{table}

\subsection{A Constrained NORTA Experiment}\label{sec:nortaexpts}
Consider a CTO (Configure To Order) manufacturer that produces three
products with uncertain demands $X_{i,t}, i=1,2,3$ for period $t$,
where $\bX_t = (X_{1,t} , X_{2,t}, X_{3,t})$ is a NORTA-generated
random vector i.i.d. over $t$. Product 1
has a net margin much higher than those of products 2 and 3. These
products share one critical component, which is procured only in a
two-month replenishment cycle for reasons of economic scaling. Taking
this into account, sales in the first period are managed to limit
fulfilling $X_{2,1}$ and/or $X_{3,1}$ demand, so that the more
profitable demand $X_{1,2}$ can be satisfied in the next period. In
particular, the demand in set $\Sx$ defined by these constraints are
unmet:
\begin{eqnarray}
  X_{2,1} + 2X_{3,1} &\geq& 3 X_{1,1}, \label{eqn:c1} \\
  X_{i,1} &\geq& U_i,\quad i=2,3. \label{eqn:ub2} 
\end{eqnarray}
Inequality (\ref{eqn:c1}) arises from the per-unit consumption of the
critical common component by the products. Constraints in
(\ref{eqn:ub2}) present upper bounds on $X_{2,1}$ and $X_{3,1}$.  


Demand is revealed at the beginning of each period, after which
fulfillment decisions $f_{i,t}$ are taken. All demand for product 1 is
filled, i.e. $f_{1,t}=X_{i,t}$. In the first period, $f_{2,1},f_{3,1}$
are optimally chosen to  minimize the lost sales:
\begin{equation}\label{eqn:exobj}
    h(\bX) = \min_{(f_{2,1}, f_{3,1}) \in \Sx}
    p_2(X_{2,1} - f_{2,1})^+ + p_3(X_{3,1} - f_{3,1})^+,
\end{equation}
\noindent where $p_2$ and $p_3$ are the prices of products 2 and 3,
respectively. Suppose $p_3=3p_2$. It is then straightforward to derive
the optimal solution $f_{3,1}^* =\min(X_{3,1},1.5 X_{1,1},U_3), f_{2,1}^*=
    \min(3X_{1,1}-2f_{3,1}^*, X_{2,1}, U_2),$ leading to $h(\bX) = X_{2,1} - f_{2,1}^* + 3(X_{3,1} - f_{3,1}^*).$ 

The sales department wishes to estimate the expected lost sales
$E[h(\bX)]$ in the first period. Suppose $U_2=\gamma\E[X_{2,1}]$
and $U_3=\gamma\E[X_{3,1}]$. The rareness of the set $\Sx$ is then
controlled by varying $\gamma$.  The NORTA vector $\bX_1$ has marginal
distributions $X_{1,t} \sim N(\mu_1=12, \sigma^2_1 = 9), X_{2,t}\sim
\text{Weibull}(\alpha=5, \beta = 10), X_{3,t}\sim \text{triang}(a=3, b
= 16, m=8)$. Two scenarios of correlations among $\bX_{1}$ are
considered: (i) positive correlation with $\rho_{12}=\quad 0.499, \,\,\rho_{13}=\quad 0.497$, and $\rho_{23}=0.747$ and (ii) negative correlation with $\rho_{12}=-0.499, \,\,\rho_{13}=-0.497$, and $\rho_{23}=0.747$. 

For the above example, it seems clear that the structure of the image of the set $\Sx$ in the Gaussian space is not easily deducible. Furthermore, neither the dominating set $\dpzs$, nor its cardinality $|\dpzs|$, are known. Such lack of adequate information leaves only the ecoNORTA estimator $\estNt$ (via the ecoNORTA algorithm) and the acceptance-rejection estimator $\estAR$ directly applicable. Accordingly, in what follows, we compare only these two estimators. 

The $\Sx$ set has one unique point closest to the origin, obtained by
solving the quadratic program $\bx^*= \min \|\bx\|$ subject to
constraints (\ref{eqn:c1}) and (\ref{eqn:ub2}). The dominating set
$\dpzsn{}$ of $\Sz$ estimated using Algorithm~1 starting
from the $\Sz$-image of $\bx^*$, illustrated in the
Figures~\ref{fig:nortaboundary_jx}, is observed to cluster around a
unique point, though far from the starting point. We set $m_k=100$ for
all $k$. Because of the stochastic nature of Algorithm~1, we repeat
the experiment 1,000 times. Figures~\ref{fig:nortaboundary_jx1} and
\ref{fig:nortaboundary_jx2} illustrate the boundaries of the feasible
region $\Sz$ by plotting samples close to or on the
boundaries. The green dots are samples close to or on the boundary
specified by the constraint on the total consumption of the common
component~(\ref{eqn:c1}). The blue and black dots show the samples
close to boundaries corresponding to the upper bounds on products 2
and 3~(\ref{eqn:ub2}).
The red dots correspond to the estimated dominating set $\dpzsn{}$ found by
Algorithm~1. We also mark the origin of the $\bZ$
space using a single red dot. We observe that
Algorithm~1 produced quite reliable estimates of
$\mathcal{Z}^*$. We also observe that samples are sparse in areas
further away from the set $\dpzsn{}$. There is a rapid
increase in the density of points close to the estimated dominating set $\dpzsn{}$,
reflecting that Algorithm~1 converges quickly.

\begin{figure}[!h]
	\centering
	\subfigure[Positive correlation]{
		\includegraphics[width=0.45\textwidth,height=0.36\textheight]{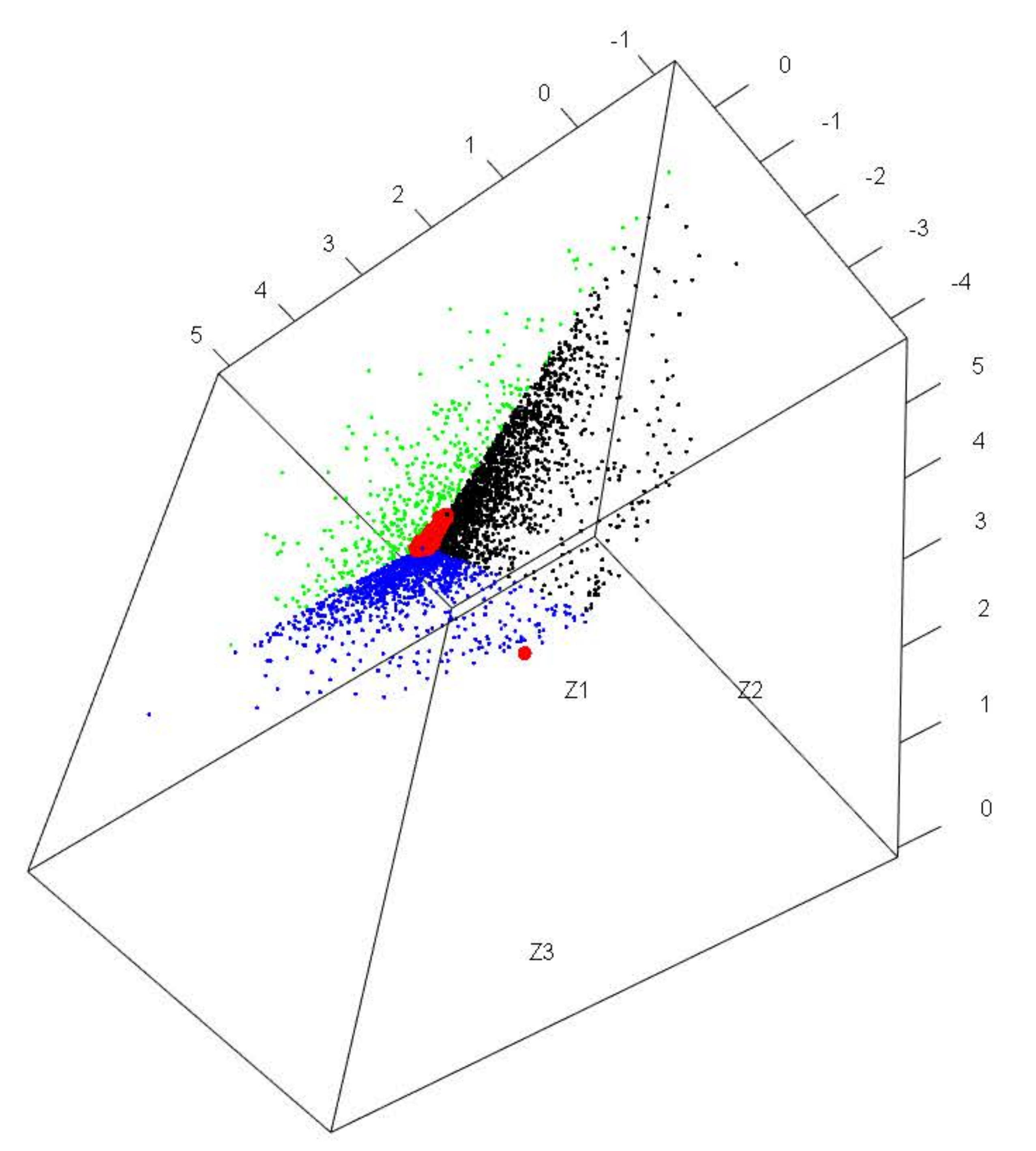}
		\label{fig:nortaboundary_jx1}
	}
	\subfigure[Negative correlation]{
		\includegraphics[width=0.48\textwidth,height=0.36\textheight]{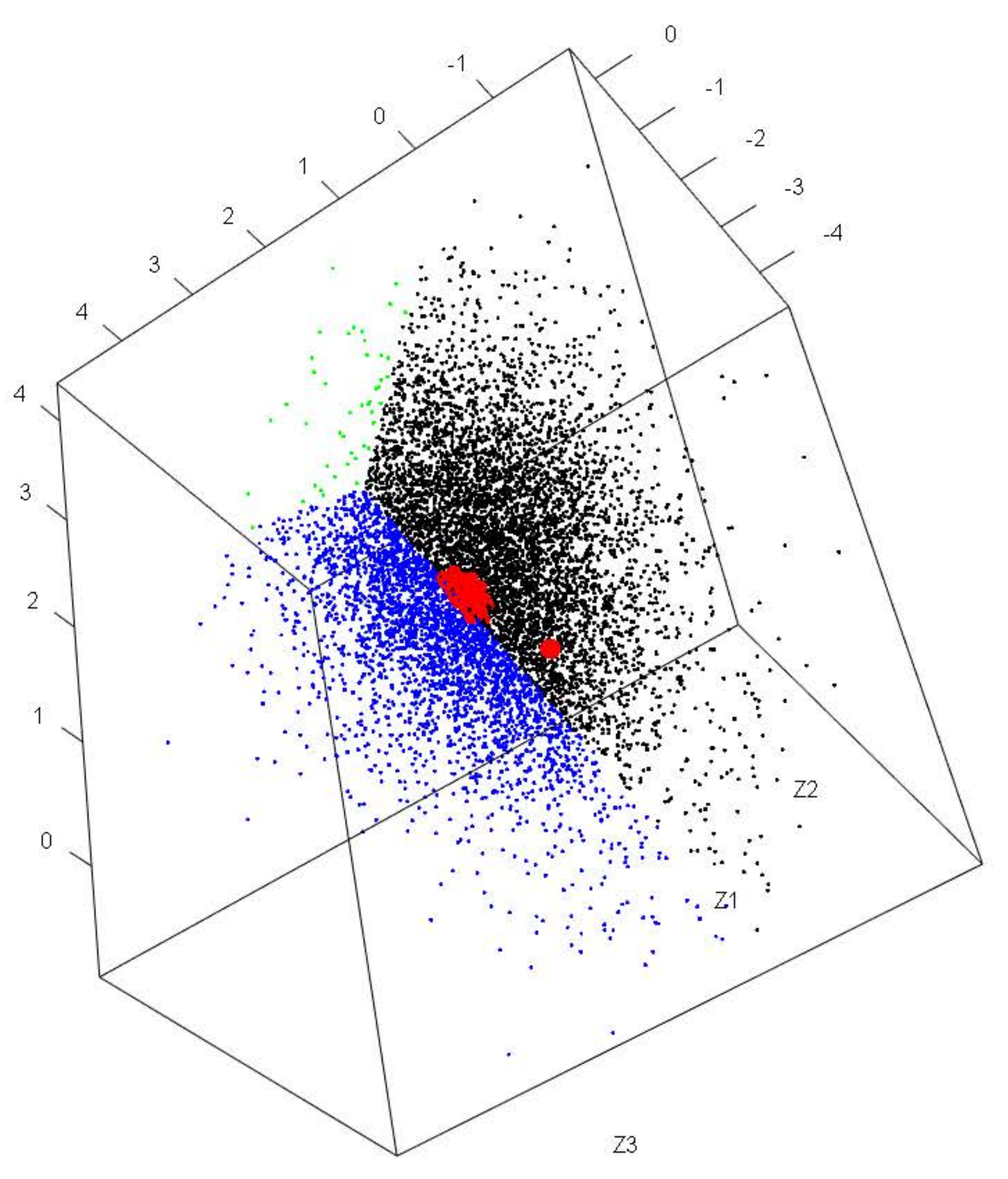}
		\label{fig:nortaboundary_jx2}
	}
	\caption{Samples close to the boundaries of the set $\Sz$ in
          the $\bZ$ space generated by Algorithm~1. In each case, we
          are looking at the origin (red dot) and the boundary of the
          set $\Sz$ from the outside. Also included are the estimated
          dominating points in a cluster of red points at the
          boundaries of the set $\Sz$.}\label{fig:nortaboundary_jx}
\end{figure}



Figure~\ref{fig:nortaboundary_jx} demonstrates how correlation
patterns affect the locations of $\dpzsn{}$ and the curvature of
$\Sz$ around $\mathcal{Z}^*$. For the positive correlation case
depicted in Figure~\ref{fig:nortaboundary_jx1}, $\mathcal{Z}^*$ seems
to be located at the corner where all three constrain surfaces intersect. In
comparison, the negative correlation between product 1's demand and the
demands for products 2 and 3 shift $\mathcal{Z}^*$ away from the
corner to a point on the ridge formed by the two constraint surfaces
specifying the upper limits on product 2 and product 3.
This is because the negative correlations mean that when demands for
product 2 and 3 are higher than average in the first period, demand
for product 1 is more likely to be lower than average. Consequently,
it is more likely for the total consumption of the common components
by products 2 and 3 to exceed the consumption by product 1, and thus
satisfy constraint~(\ref{eqn:c1}).


Using the estimated set $\mathcal{Z}^*$, we compare the performance of
$\estNt$ against the acceptance-rejection estimator $\estAR$ in Table~\ref{tab:allscn}. We
experimented with four $\gamma$ values that make
$p \defn \P(\bX\in\Sx)$ range from $10^{-3}$ to $10^{-6}$ for
the positive correlation scenario. As discussed previously, the
negative correlation case has higher probabilities for these $\gamma$
values. The estimator $\estNt$ outperforms $\hat{\alpha}_{\cal
{AR}}$ as expected. The performance of $\estNt$ is
affected by the curvature of $\Sz$ at the dominating point. For
instance, compare the negative correlation case with $\gamma = 1.771$
to the positive correlation case with $\gamma = 1.7300$. While the
former has a smaller probability, the relative error is actually much
smaller, due to the different curvature at the dominating point. 

\begin{table}[!h]
	\centering
	\begin{tabular}{c||c|c|c|c||c|c|c|c}
		\hline
		& \multicolumn{4}{c||}{Positive Correlation} & \multicolumn{4}{c}{Negative Correlation} \\ \hline  
		$\gamma$ & 1.6300 & 1.7300 & 1.7607 & 1.771 & 1.6300 & 1.7300 & 1.7607 & 1.771 \\		\hline
		$p$ & 1.00e-3 & 1.01e-4 & 1.02e-5& 1.08e-6 & 9.81e-3 & 1.31e-3 & 2.01e-4& 3.43e-5 \\ \hline
		
		$\alpha$ & 7.92e-4 & 8.75e-5 & 1.07e-5 & 1.23e-6 & 1.01e-2 & 1.56e-3 & 3.00e-4 & 6.38e-5 \\ \hline
		
		$\E[\hat{p}_{\cal {AR}}^2]/p^2$ & 986 & 9.81e3 & 1.17e5 & 1.37e6 & 1.01e2 & 7.74e2 & 5.04e3 & 2.86e4 \\ \hline
		
			
		$\E[\hat{p}^2_{\rm\tiny eN}]/p^2$ & 12.6 & 13.2 & 13.9 & 18.4 & 3.44 & 2.74 & 2.25 & 2.07 \\ \hline
		
		$\E[\estAR^2]/\alpha^2$ & 1.56e3 & 1.56e4 & 1.89e5 & 3.27e6 & 1.56e2 & 1.11e3 & 6.67e3 & 3.53e4 \\ \hline
		
				
		$\E[\estNt^2]/\alpha^2$ & 17.3 & 20.3 & 19.6 & 29.3 & 4.70 & 3.77 &  2.90 & 2.45 \\ \hline
	\end{tabular}
	\caption{Performance of the acceptance-rejection estimator $\estAR$ and the ecoNORTA estimator $\estNt$ (obtained through the ecoNORTA algorithm) on a configure to order manufacturing problem.}\label{tab:allscn}
\end{table}

\bibliographystyle{plainnat}
\bibliography{stochastic_optimization,rvgeneration,extra}




\section{Appendix}\label{techproofs} 

\ProofOf{\refthm{thm:nonexistence}}
  We first notice that
  $\alpha(z_1^*) \defn \int_{z_1^*}^{\infty} \phi (z) dz $
  satisfies the well-known asymptotic \begin{equation}\label{alpha1drate}\alpha(z_1^*) \sim \phi(z_1^*)/z_1^*\end{equation} as $z_1^* \to \infty$. The assertion in (a) follows upon noticing that $\E[\estAR^2(z_1^*)] = \E[\estAR(z_1^*)] = \alpha(z_1^*)$ and that $\E[\estAR^2(z_1^*)] \sim \alpha^{-1}(z_1^*)$ as argued in Section \ref{sec:natural}.

To obtain the result in (b), write
\begin{align}\label{exptilt}
\E[\estNml^2(z_1^*)]&=\int_{z_1^*}^{\infty} \left(\frac {\phi(z)}{\phi(z;{\mu})}\right)^2 \phi(z;{\mu})dz
=\frac{1}{2\pi}\exp\{\mu^2\}\int_{z_1^*}^{\infty} \exp\{-\frac{1}{2}(z+\mu)^2\}dz\nonumber\\
&\sim\frac{1}{2\pi}{z_1^*}^{-1}\exp\{\mu^2 - \frac{1}{2}(z_1^*+\mu)^2\}
\end{align} and use the expression in (\ref{alpha1drate}).
Finally, it is easily seen that the relative error $\E[\estNml^2(z_1^*)] / \alpha^2(z_1^*)$ is minimized by setting
$\mu=z_1^*+o(z_1^*)$ in the expression for $\E[\estNml^2(z_1^*)] / \alpha^2(z_1^*)$.
\qed

\ProofOf{\refthm{thm:mserate}} Since $W$ is a random variable having the density $f_{\lambda(z_1^*)}(x)$ and $\lambda = \theta z_1^*$, we can write \begin{align}\label{optimal} \E[\estOpt^2]  & =  (\frac{1}{2\pi})^d\int_{{z_1}^*}^{\infty} {{z_1}^*}^{2p}\frac{\exp\{-{w}^2\}}{\theta z_1^* \exp\{-\theta z_1^*(w-z_1^*)\}}\eta^2 (w-z_1^*)^{2s} \,dw \nonumber \\
& =  (\frac{1}{2\pi})^d\int_{{z_1}^*}^{\infty} {{z_1}^*}^{2p}\frac{1}{\theta {z_1}^*}\exp\{-{w}^2 + \theta {z_1}^*(w-{z_1}^*)\}\eta^2 (w-z_1^*)^{2s} \,dw \nonumber \\
& =  (\frac{1}{2\pi})^d{\theta {z_1}^*}^{-1} \eta^2 \int_{{z_1}^*}^{\infty} {{z_1}^*}^{2p}\exp\{-{{z_1}^*}^2 - (2-\theta){z_1}^*(w-{z_1}^*) - (w-{z_1}^*)^2\} (w-z_1^*)^{2s} \,dw \nonumber \\
& = (\frac{1}{2\pi})^d\theta^{-1} (2-\theta)^{-2s-1}{{z_1}^*}^{2p-2s-2}\eta^2\exp\{-{{z_1}^*}^2\}\int_{0}^{\infty} \exp\{-u - \frac{u^2}{(2-\theta)^2{{z_1}^*}^2}\} u^{2s} \,du \nonumber \\
& \sim (\frac{1}{2\pi})^d{\theta}^{-1}(2-\theta)^{-2s-1}\Gamma(2s+1)\eta^2{{z_1}^*}^{2p-2s-2}\exp\{-{{z_1}^*}^{2}\},
\end{align} where the last line follows from the application of Lebesgue's dominated convergence theorem~\citep{bil95} and the definition of the gamma function. 

To prove $L.2$, we first see that \begin{align} \E[\estOpt] &= (\frac{1}{2\pi})^{d/2}\int_{{z_1}^*}^{\infty} \exp\{-\frac{1}{2}w^2\}w^p\eta (w-z_1^*)^s \label{expest1} \\
& \sim (\frac{1}{2\pi})^{d/2}\Gamma(s+1){{z_1}^*}^{p-1-s}\eta\exp\{-\frac{1}{2}{{z_1}^*}^2\}, \label{expest2}\end{align} where the last asymptotic equivalence (as $z_1^* \to \infty$) follows from Lemma \ref{lem:assrates}. We also recall, using the notation from (\ref{alpharate}), that  \begin{align} \alpha = (\frac{1}{2\pi})^{d/2}\int_{{z_1}^*}^{\infty} \exp\{-\frac{1}{2}w^2\}w^p\tilde{v}_p(w-z_1^*)^s  + \sum_{i=1}^{p-1}(\frac{1}{2\pi})^{d/2}\int_{{z_1}^*}^{\infty} \exp\{-\frac{1}{2}w^2\}w^p\tilde{v}_i(w-z_1^*)^s, \label{alpharateagain}\end{align} where as $t \to 0$ \begin{align}\label{vtildeasymp}\tilde{v}_p(t) = \eta(t)^s + o(t^s).\end{align} We see that (\ref{vtildeasymp}) implies that for given $\epsilon >0$, there exists $\delta(\epsilon) > 0$ (not dependent on $z_1^*)$ such that for all $|z_1 -{z_1}^*| \leq \delta(\epsilon)$, \begin{align}\label{vtildeasympexplicit}|\frac{\tilde{v}_p(z_1-z_1^*) - \eta(z_1-z_1^*)^s}{\eta(z_1-z_1^*)^s}| \leq \epsilon.\end{align} Combining (\ref{expest1}) and (\ref{alpharateagain}), we see that \begin{align}\label{biasbound} |\E[\estOpt] - \alpha| &\leq (\frac{1}{2\pi})^{d/2} \int_{w - {z_1}^* \leq \delta(\epsilon)}^{\infty} \exp\{-\frac{1}{2}w^2\}w^p|\eta (w-z_1^*)^s - \tilde{v}_p(w-{z_1}^*)| \nonumber \\ & \hspace{0.2in} + (\frac{1}{2\pi})^{d/2} \int_{w - {z_1}^* > \delta(\epsilon)}^{\infty} \exp\{-\frac{1}{2}w^2\}w^p|\eta (w-z_1^*)^s - \tilde{v}(w-{z_1}^*)| \nonumber \\
& \hspace{0.2in} + \sum_{i=1}^{p-1}(\frac{1}{2\pi})^{d/2}\int_{{z_1}^*}^{\infty} \exp\{-\frac{1}{2}w^2\}w^i\tilde{v}_i(w-z_1^*)^s \nonumber \\
& \triangleq I_1 + I_2 + I_3.\end{align} Towards bounding $I_1$ appearing in (\ref{biasbound}), we use (\ref{vtildeasympexplicit}) to write for large enough $z_1^*$ that \begin{align}\label{I1bound} I_1 & \leq (\frac{1}{2\pi})^{d/2} \epsilon\int_{w - {z_1}^* \leq \delta(\epsilon)}^{\infty} \exp\{-\frac{1}{2}w^2\}w^p \eta (w-z_1^*)^s \nonumber \\ & \leq (\frac{1}{2\pi})^{d/2} \epsilon\int_{{z_1}^*}^{\infty} \exp\{-\frac{1}{2}w^2\}w^p \eta (w-z_1^*)^s \nonumber \\ &\sim \epsilon \alpha,\end{align} where the last step follows from the application of Lemma\ref{lem:assrates} and from the assertion of Theorem\,\ref{thm:alpharate}. To bound the second integral in (\ref{biasbound}), we note that \begin{align}\label{I2bound} I_2 & \leq (\frac{1}{2\pi})^{d/2} \int_{w - {z_1}^* > \delta(\epsilon)}^{\infty} \exp\{-\frac{1}{2}w^2\}w^p(\eta (w-z_1^*)^s + \tilde{v}_p(w-{z_1}^*)) \nonumber \\ & \sim 2(\frac{1}{2\pi})^{d/2}\Gamma(s+1){({z_1}^* + \delta(\epsilon))}^{p-1-s}\eta\exp\{-\frac{1}{2}{({z_1}^* + \delta(\epsilon))}^2\}, \nonumber \\ & = o(\alpha),\end{align} where the last step again follows from the application of Lemma\ref{lem:assrates} and from the assertion of Theorem~\ref{thm:alpharate}. From arguments in the proof of Theorem~\ref{thm:alpharate}, we know that each summand in the expression for $I_2$ is $o(\alpha)$ giving us that $I_3 = o(\alpha)$ as ${z_1}^* \to \infty$. This last observation along with (\ref{I1bound}) and (\ref{I2bound}), and observing that $\epsilon$ is arbitrary, proves the assertion in $L.2$.

To prove $L.3$, we notice after some algebra that \begin{align}\label{mse_split}\frac{\E[(\estOpt - \alpha)^2]}{\alpha^2} = \frac{\E(\estOpt^2)}{\alpha^2} - 1 + \frac{\E[\estOpt - \alpha]}{\alpha}.\end{align}
We know from the assertion in $L.2$ that the last term in (\ref{mse_split}) is $o(1)$. We also know from the assertion in $L.1$ and Theorem \ref{thm:alpharate} that the first term in (\ref{mse_split}) converges to ${\theta}^{-1}(2-\theta)^{-2s-1}\Gamma(2s+1)\Gamma^{-2}(s+1).$ We conclude from these arguments that the assertion in $L.3$ holds. \qed

\ProofOf{\reflem{expRE}} 
The point $z^*_1$ represents the nearest point of the set
$\Szcom{z^*_1}$ to the origin along the $z_1-$axis. Recall the volume
expansion defined in~\refasm{curvcoeffknown}:
\begin{equation}\label{eq:volexpinc}
v_p(t, \bc) = \gamma_p \int \ind{\Szcom{z^*_1}\,|\,t} \prod_{i=2}^d
z_i^{c_i} \, \,dz_i = \gamma_p \eta(\bc) t^{s(\bc)} \,\,+\,\, o(t^{s(\bc)}),
\end{equation}
 where $t=(z_1-z_1^*)$, $\bc = (c_2,\ldots,c_d)$, the cross-section
 set $\{\Szcom{z^*_1}\,|\,z_1 = z^*_1 + t\}$ 
, and the argument $\bc$ has been included for emphasis.

Part (i): For the translation regime $\Szcom{z^*_1} = \{ (z_1,\bz) \,|\,
(z_1-z_1^*, \bz) \in \Szcom{0} \}$, the set $\{\Szcom{z^*_1}\,|\,t\} 
= \{\bz \,|\, f(\bz) \le t\}$ remains the same as $z^*_1\tndi$, and
 the volume expansion follows~\refeq{eq:volexpinc}. 
\refthm{thm:alpharate} gives us that
\begin{equation}\label{eq:exptrsalp}
\alpha(z^*_1) \quad\sim\quad
\frac{1}{(2\pi)^{d/2}}\,\,\gamma_p\,\,\eta(\bc)\,\,\Gamma(s(\bc)+1)\quad
{(z_1^*)}^{p-1-s(\bc)}\exp\{-\frac{1}{2}{z_1^*}^2\}.  
\end{equation}
The second moment of $\estExp$ is
\begin{eqnarray}
\E[\estExp^2(z^*_1)] &=& 
\int_{\Szcom{z^*_1}} \left( \sum_{c_1=1}^p \gamma_{c_1}
z_1^{2c_1} \prod_{i=2}^d
z_i^{2c_i(c_1)}\right)^2\,\,\,\frac{\phi^2(z_1)}{f_{z^*_1}(z_1)} \left( \prod_{i=2}^d \phi(z_i)\,\right) \,\,
d\mathbf{V} \,\,dz_1 \nonumber\\
&\sim& 
\gamma_p^2 \frac 1 {(2\pi)^{(d-1)/2}} \,\,\, \int_{z_1^*}^{\infty} z_1^{2p}
\,\frac{\phi^2(z_1)}{f_{z^*_1}(z_1)} v_p(z_1-z_1^*, 2\bc) \,dz_1 \label{eq:secmomstep}\\
&=& \gamma_p^2 \frac 1 {(2\pi)^{(d+1)/2}}\,\, \frac {\eta(2\bc)}{z_1^*} \, \int_{z_1^*}^{\infty} z_1^{2p}
 \exp\left(-z^2_1+z^*_1(z_1-z^*_1)\,\,\right) \,\,\,\,
 (z_1-z_1^*)^{s(2\bc)} \,dz_1. \nonumber
\end{eqnarray}
Use the substitution $v = (z^*_1) (z_1-z_1^*)$, and
apply Lemma\ref{lem:assrates} (i) to get
\[
\E[\estExp^2(z^*_1)]\quad \sim\quad \gamma_p^2 \frac 1
{(2\pi)^{(d+1)/2}}\,\, \eta(2\bc) \,\,\Gamma(s(2\bc)+1)\quad
(z_1^*)^{2p-2-s(2\bc)}\,\, e^{-(z^*_1)^2}.  
\]
Comparing against $\alpha^2(z^*_1)$ from~(\refeq{eq:exptrsalp}) gives us
the desired growth rate of the relative error of $\estExp(z^*_1)$ as:
\[
\frac {\E[\estExp^2(z^*_1)]} {\alpha^2(z^*_1)}
\quad\sim\quad \left(\frac
{(2\pi)^{(d-1)/2}\,\,\eta(2\bc)\,\,\Gamma(s(2\bc)+1) }
{\quad\,\,\eta^2(\bc)\,\,\Gamma^2(s(\bc)+1)} \right) \quad (z^*_1)^{2s(\bc)-s(2\bc)}. 
\]
 
Part (ii): The key step, as in the analysis above, is in
characterizing the volume expansion $v_p(t,\bc)$ for the scaled set
$\Szcom{z^*_1}$. Using the observation that the cross-sectional set 
$\{\Szcom{z^*_1} \,|\, t \} = \{ \bz
\,|\, f ({\bz} / z^*_1 )\le t \}$,
substitute $u_i = z_i / z^*_1$ for $i=2,\ldots,d$, to get
\begin{eqnarray*}
v_p(z_1 - z^*_1, \bc)\quad \sim\quad \eta(\bc) \, (z^*_1)^{-\sum_{i=2}^d (c_i+1)} (z_1-
z_1^*)^{s(\bc)},
\end{eqnarray*}
where $\eta(\bc)$ and $s(\bc)$ are the parameters of the volume
expansion of the original set $\Szcom{0}$ as defined
in~\refeq{eq:volexpinc}. Thus, $\alpha(z^*_1)$ in this case follows
from~\refeq{eq:exptrsalp} as 
\begin{equation}
\alpha(z^*_1) \quad\sim\quad
\frac{1}{(2\pi)^{d/2}}\,\,\gamma_p\,\,\eta(\bc)\,\,\Gamma(s(\bc)+1)\quad
{(z_1^*)}^{p-1-s(\bc)}\,\,{(z_1^*)}^{-\sum_{i=2}^d (c_i+1)}\,\,\exp\{-\frac{1}{2}{z_1^*}^2\}.  
\label{eq:expsclalp}
\end{equation}
The second moment on the other hand uses $v_p(z_1-z^*_1,2\bc)$ and so is 
\begin{equation}
\E[\estExp^2(z^*_1)]\quad \sim\quad \gamma_p^2 \frac 1
{(2\pi)^{(d+1)/2}}\,\, \eta(2\bc) \,\,\Gamma(s(2\bc)+1)\quad
(z_1^*)^{2p-2-s(2\bc)}\,\,{(z_1^*)}^{-\sum_{i=2}^d (2c_i+1)}\,\, e^{-(z^*_1)^2}.  
\label{eq:expsclerr}
\end{equation}
Comparing~\refeq{eq:expsclerr} with~\refeq{eq:expsclalp} gives us 
\[
\frac {\E[\estExp^2(z^*_1)]} {\alpha^2(z^*_1)}
\quad\sim\quad \left(\frac
{(2\pi)^{(d-1)/2}\,\,\eta(2\bc)\,\,\Gamma(s(2\bc)+1) }
{\quad\,\,\eta^2(\bc)\,\,\Gamma^2(s(\bc)+1)} \right) \quad
(z^*_1)^{2s(\bc)-s(2\bc)\,\,+\,\,(d-1)}. \qquad \ProofEnd
\]

\subsection{The Exponential-Twisting Estimator $\estNml$}\label{sec:exptwist}

How does the performance of the partial information estimators $\estExp$ and $\estLap$ compare to the
exponential-twisting estimator $\estNml$? The following result
provides the answer.

\begin{lemma}\label{lem:estNmlmulti}
Let $\estNml(z^*_1)$ estimate $\alpha (z^*_1) = \E[ h(z_1,\bz)
  \,\,\mathbb{I}\{ (z_1,\bz) \in \Szcom{z^*_1} \}]$, where $$h(z_1,\bz)
= \sum_{c_1=1}^p \gamma_{c_1} z_1^{c_1}\prod_{i=2}^d z_i^{c_i(c_1)}$$
  and all
$c_i>0$. Let $\xi = \left(\frac { (2\pi)^{d/2}
  \,\,\eta(2\bc)\,\,\Gamma(s(2\bc)+1)} 
{2^{s(2c)+1}\,\,\,\eta^2(\bc)\,\,\Gamma^2(s(\bc)+1)} \right).$
Then
\[
\begin{array}{lrcl}
\mbox{(i) For the translated set }\Szcom{z^*_1}\mbox{, as }z^*_1\tndi,
& \frac {\E[\estNml^2(z^*_1)]} {\alpha^2(z^*_1)}
&\sim& \,\,\xi\,\,\,
  (z^*_1)^{2s(\bc) \,\,-\,\, s(2\bc)\,\, +\,\,1},\quad\mbox{ and} \\
\mbox{(ii) For the scaled set }\Szcom{z^*_1}\mbox{ as  }z^*_1\tndi,
& \frac {\E[\estNml^2(z^*_1)]} {\alpha^2(z^*_1)}
&\sim& \,\,\xi\,\,\,
(z^*_1)^{2s(\bc) \,\,- \,\,s(2\bc)\,\,+1\,\, +\,\, (d-1) }\,\,\,.
\end{array}
\]
\end{lemma}
Thus, for both the translation and scaling regime, the
RE($\estNml$)/RE($\estLap$) $= O((z^*_1)^{1/2}) \tndi$ with $z^*_1$.

\ProofOf{\reflem{lem:estNmlmulti}}
We shall show (i), and the proof of (ii) follows similarly as in the proof
of \reflem{lem:szprop}.  The $\alpha(z^*_1)$ for the translation
regime is given by~\refeq{eq:exptrsalp}, so we need to calculate the
second moment of $\estNml(z^*_1)$. Starting from the
step~\refeq{eq:secmomstep} above:
\begin{eqnarray*}
\E[\estExp^2(z^*_1)]&\sim& 
\frac {\gamma_p^2 \,\,\eta(2\bc)} {(2\pi)^{d/2}}\,\,
 \int_{z_1^*}^{\infty} z_1^{2p} 
 \exp\left(-z^2_1+\frac {(z_1-z^*_1)^2} 2\,\,\right) \,\,\,\,
 (z_1-z_1^*)^{s(2\bc)} \,dz_1 \\
&=& 
\frac {\gamma_p^2 \,\,\eta(2\bc)} {(2\pi)^{d/2}}\,\,
e^{-(z^*_1)^2}\,\, \int_{0}^{\infty} (t+z^*_1)^{2p}\,\exp\{-\frac
{t^2} 2 - 2tz^*_1\}\,t^{s(2\bc)}\,\,dt \\
&=&  
\frac {\gamma_p^2 \,\,\eta(2\bc)} {(2\pi)^{d/2}}\,\,
e^{-(z^*_1)^2}\,\,(z^*_1)^{2p-s(2\bc)-1}\,\, \int_{0}^{\infty} \left(\frac
u {(z^*_1)^2} +1 \right)^{2p}\,\exp\{-\frac
{u^2} {2(z^*_1)^2} - 2u\}\,\,\,u^{s(2\bc)}\,\,du \\
&\sim& 
\frac {\gamma_p^2 \,\,\eta(2\bc)} {(2\pi)^{d/2}}\,\,
e^{-(z^*_1)^2}\,\,(z^*_1)^{2p-s(2\bc)-1}\,\, \int_{0}^{\infty} 
\exp\{-2u\}\,\,u^{s(2\bc)}\,du\\
&=& 
\frac {\gamma_p^2 \,\,\eta(2\bc)} {(2\pi)^{d/2}}\,\,
e^{-(z^*_1)^2}\,\,(z^*_1)^{2p-s(2\bc)-1}\,\, 
\frac {\Gamma(s(2\bc)+1)} {2^{s(2\bc)+1}}.
\end{eqnarray*}
where the second equality substitutes $t=(z_1-z^*_1)$, the third
substitutes $u =t z^*_1$, the fourth follows from applying the
dominated convergence theorem and the last uses the gamma function
definition. Comparing this with the $\alpha(z^*_1)$ gives us the
desired result.
\ProofEnd

\subsection{Properties of the NORTA Transformation $\bT$}\label{sec:tprop}
This section provides some structural properties of the set $\Sz$
given as the pre-image of the feasible set $\Sx$ under the NORTA
transformation $\bT$.

\begin{lemma} \label{lem:szprop}
Suppose the constrained sets are defined as
$\Sx = \{\bx \,\,:\,\,\,l_i(\bx)\ge 0 ,\,\forall i=1,\ldots,m \}\quad$ 
and $\Sz=\{\bz\,\,:\,\,\,l_i(\bT(\bz))\ge 0,\,\forall i=1,\ldots,m\}$, 
where each constraint $l_i$ is continuous in $\bx$.  
Further, suppose
the set $\Sx$ is compact.  Then:
  \begin{itemize}
  \item[(i)] $\Sz$ is compact.
  \item[(ii)] The boundary $\partial\Sx$ of set $\Sx$ maps to points on the
    boundary $\partial\Sz$ of $\Sz$.
    \item[(iii)] $\Sz$ has an interior if $\Sx$ is full-dimensional. 
  \item[(iv)] Further, if $\Sx$ is connected, then so is $\Sz$.
  \end{itemize}
\end{lemma}
\ProofOf{\reflem{lem:szprop}}
Property (i) follows from the continuity and one-to-one onto nature of
the maps $T_i$ and its inverse $T_i^{-1}$.  For (ii), note that since
$\Sx$ is closed, there exist sequences in its complement $\Sx^c$ that
converge to $x^{\delta}$, which map to sequences in $\mathcal{S}^c_{z_1^*,\bz}$. Further,
these converge to $\bT^{-1}(x^{\delta})$ and hence
$\bT^{-1}(x^{\delta})\in\partial\Sz.$

For (iii), consider any point $x^o$ in the strict interior of $\Sx$. For any
vector $v$, there exists an $\epsilon>0$ such that $\forall
t\,\in(-\epsilon,\epsilon),\,\,\, l_i(x^o+tv)\ge 0 ,\,\forall
\,i=1,\ldots,m.$ For any vector $u$, taking a Taylor approximation,
the point $z^o+\epsilon u$ satisfies $l_i(\bT(z^o + \epsilon
u))\,\,\approx \,\,l_i\left(\bT(z^o) + \epsilon v\right) = l_i(x^o +
\epsilon v) \ge 0$ for a sufficiently small $\epsilon>0$, where $v$
has components $v_i = u_i \,\,\partial T_i(z^o_i)/\partial z_i$, where the
derivative $\partial T_i(z_i)/\partial z_i$ exists and is non-zero
everywhere because $F_i$ is continuous, strictly increasing and has a
non-zero density.

In (iv), for any two points $x^1, x^2 \in \Sx$, there exists at least
one path between them. By definition, the pre-image of points in $\Sx$
are in $\Sz$. Further, since each point on a path is a limit of a
sequence of points on the path, the image of the path is itself a
path, giving us (iv).  \ProofEnd

\subsection{Additional Examples}\label{additional_examples} 

This section provides two examples. The first shows how the
parameters $\eta, s$ can be estimated for sets with polynomial
boundaries. The parameter $s$ is guessed based
on the polynomial terms involved.
The second example continues the Gaussian numerical experiment
presented in~\refsec{sec:mvnexpts} and provides additional details for
the results presented.

\begin{example} Suppose that $\alpha(z_1^*) = \mathbb{E}\{(Z_1^aZ_2^b Z_3^c) \mathbb{I}\{(Z_1,Z_2,Z_3) \in \Sz\},$ where $(Z_1,Z_2,Z_3)$ is the standard trivariate normal and the set $\Sz$ is given by $\Sz := \{(z_1,z_2,z_3) \in \real^3: z_1 \geq z_1^* + z_2^2 + z_3^4\}.$ Denoting $\Sz(z_1-z_1^*) := \{(z_2,z_3): z_2^2 + z_3^4 \leq z_1 - z_1^*\}$, we see then that \begin{align}\label{mvsmoothex1}\alpha(z_1^*) & = \frac{1}{(2\pi)^{3/2}}\int_{z_1^*}^{\infty} z_1^a \exp\{-\frac{1}{2}z_1^2\}  \int_{\Sz(z_1 - z_1^*)} \, z_2^bz_3^c\exp\{-\frac{1}{2}(z_2^2 + z_3^2)\} \nonumber \\
& = \frac{1}{(2\pi)^{3/2}}\int_{z_1^*}^{\infty} z_1^a\exp\{-\frac{1}{2}z_1^2\} g(z_1 - z_1^*),
\end{align} where $$g(z_1-z_1^*) := (z_1 - z_1^*)^{\frac{b}{2} + \frac{c}{4}}\int_{\Sz(1)} \tilde{z}_2^b\tilde{z}_3^c \exp\{-\frac{1}{2}((z_1 - z_1^*)\tilde{z}_2^2 + (z_1 - z_1^*)^{1/2}\tilde{z}_3^2)\}$$ upon making the substitution $\tilde{z}_2 = z_2 (z_1 - z_1^*)^{-1/2}$, $\tilde{z}_3 = z_3 (z_1 - z_1^*)^{-1/4}$. We thus see that as $t \to 0$, $$g(t) = \eta t^{\frac{b}{2} + \frac{c}{4}} + o(t), \quad \eta = \int_{\Sz(1)} \tilde{z}_2^b\tilde{z}_3^c.$$   Then, to calculate the exact rate at which $\alpha(z_1^*) \to 0$ as $z_1^* \to \infty$, we invoke Lemma \ref{lem:assrates} to obtain $$\alpha(z_1^*) \sim \frac{\eta}{(2\pi)^{3/2}}\Gamma(\frac{b}{2} + \frac{c}{4} + 1){z_1^*}^{a - \frac{b}{2} - \frac{c}{4} - 1}\exp\{-\frac{1}{2}{z_1^*}^2\}.$$
\end{example}

\noindent {Experiment in~\refsec{sec:mvnexpts}
  continued.}  We develop an exact, recursive expression for
$\alpha(z_1^*)$, which is used in calculating the results presented in
Section~\ref{sec:mvnexpts} for the multivariate Gaussian problem
presented there. If we define $\alpha_d(z^*) \defn \P\{\mathbf{Z}
\in \cS_{\bz}\}$, where $\mathbf{Z}$ is a $d$-dimensional standard
Gaussian random vector, then \[ \alpha_d(z^*)=\alpha_{d-1}(z^*) -
c_d\sum_{i=0}^{\frac{d-1}{2}-1}\binom{\frac{d-1}{2}-1}{i}\left(-z^*-\frac{1}{2}\right)^{\frac{d-1}{2}-1-i}\int_{z^*+1/2}^{\infty}\frac{e^{-t^2/2}t^i}{\sqrt{2\pi}}\;dt
\] 
where $c_d=\frac{e^{\frac{1}{8}+ \frac{z^*}{2}
}}{\Gamma \left(\frac{d-1}{2}+1\right) 2^{\frac{d-1}{2}-1}}, \quad
\alpha_3(x^*)=\overline\Phi(x^*)-e^{\frac{1}{8}+ \frac{x^*}{2}
}\overline\Phi\left(x^*+\frac{1}{2}\right), $  and
\[\alpha_2(x^*)=\frac{1}{\sqrt{2\pi}}\int_{x^*}^{\infty}e^{-t^2}(1-2\overline\Phi\left(\sqrt{t-x^*}\right))\;dt.
\]

\end{document}